\documentclass{eptcs}
\usepackage{mathtools,mathpartir,amsmath,amssymb}
\usepackage{hyperref}
\usepackage{keyval}
\usepackage{prooftree}
\usepackage[all]{xy}
\usepackage{dsfont}
\usepackage{enumerate}
\usepackage{color}
\usepackage[dvipsnames]{xcolor}
\usepackage{tikz}                   
\usetikzlibrary{shadows,arrows,shapes,automata,positioning,decorations.markings,fit}
\usepgflibrary{arrows.spaced}
\usepackage{wrapfig}
\usepackage{breakurl}
\usepackage[misc,geometry]{ifsym}
\usepackage{orcidlink}

\title{Modular Multiparty Sessions  with Mixed  Choice}

\def\titlerunning{Modular Sessions  with Mixed  Choice} 

\def\authorrunning{ 
Barbanera \& Dezani-Ciancaglini
}

\author{Franco Barbanera
\thanks{Partially supported by 
Project ``National Center for HPC, Big Data e Quantum Computing",  Programma M4C2, Investimento 1.3 – Next Generation EU; and by the 
PIAno di inCEntivi per la RIcerca di Ateneo 2024-2026 UniCT (Linea di Intervento 1).
}
\institute{
Dipartimento di Matematica e Informatica,
Universit\`a di  Catania, Catania, Italy}
\email{franco.barbanera@unict.it}
\and
Mariangiola Dezani-Ciancaglini 
\institute{
Dipartimento di Informatica,
Universit\`a di Torino, Torino, Italy}
\email{dezani@di.unito.it}
}

\newtheorem{definition}{Definition}[section]
\newtheorem{lemma}[definition]{Lemma}

\newtheorem{theorem}[definition]{Theorem}

\newtheorem{example}[definition]{Example}

\newenvironment{proof}{{\em Proof.}}{
}

\newcommand{\set}[1]{\{#1\}}
\newcommand{\rn}[1]{\textsc{\footnotesize [{#1}]}}
\newcommand{\coDef}{::=^{coind}}

\newcommand{\G}{{\sf G}}

\newcommand{\End}{\ensuremath{\mathtt{End}}}

\newcommand{\inact}{\ensuremath{\mathbf{0}}}
\newcommand{\PP}{P}
\newcommand{\Q}{Q}

\newcommand{\D}{Q}

\newcommand{\ptp}[1]{
  \ensuremath{\mathsf{\color{blue}{ #1}}}
}
\newcommand{\msg}[1]{\mathit{\color{BrickRed}{#1}}}

\newcommand{\pa}{\ptp{a}}
\newcommand{\pb}{\ptp{b}}
\newcommand{\pc}{\ptp{c}}
\newcommand{\pd}{\ptp{d}}
\newcommand{\pe}{\ptp{e}}
\newcommand{\px}{\ptp{x}}
\newcommand{\pS}{\ptp{s}}
\newcommand{\pgS}{\ptp{gs}}
\newcommand{\pp}{\ptp{p}}
\newcommand{\q}{\ptp{q}}
\newcommand{\pq}{\ptp{q}}
\newcommand{\pr}{\ptp{r}}
\newcommand{\ps}{\ptp{s}}
\newcommand{\pt}{\ptp{t}}
\newcommand{\pu}{\ptp{u}}
\newcommand{\pv}{\ptp{v}}

\newcommand{\pw}{\ptp{w}}

\newcommand{\PS}{S}
\newcommand{\PQ}{Q}
\newcommand{\PR}{R}

\newcommand{\pplus}{\hspace{2pt}\scalebox{0.7}{\raisebox{1mm}{.....}\hspace{-2.6mm}\raisebox{-1.2mm}{\rotatebox{90}{.....}}}\hspace{6pt}}

\newcommand{\ppplus}{\hspace{2pt}\scalebox{0.7}{\raisebox{1mm}{.....}\hspace{-2.7mm}\raisebox{-1.2mm}{\rotatebox{90}{.....}}}\hspace{6pt}}

\newcommand{\SLTS}[1]{\xrightarrow{#1}}
\newcommand{\Act}[3]{#1 #2 #3}

\newcommand{\qed}{\hfill $\Box$}

\newcommand{\derP}{\vdash^{\mathcal{P}}}

\newcommand{\pP}[2]{\ensuremath{#1\text{\bf [}#2\text{\bf ]}}}
 \newcommand{\parN}{\mathrel{\|}}
 
 \newcommand{\tyn}[2]{#1\vdash^{\mathcal P} #2}

 \newcommand{\NamedCoRule}[5][]{\IInfer[#1]{#2}{ #3 }{#4}{#5}} 
\newcommand {\IInfer} [5] [] {
  \inferrule*[%
    fraction={===}, 
    left={\textsc{#2}},%
    right={$\begin{array}{l} #5 \end{array}$}, 
    #1
  ]%
   {#3}{#4}}
  \newcommand{\NamedRule}[5][]{\Infer[#1]{#2}{ #3 }{#4}{#5}} 
\newcommand {\Infer} [5] [] {
  \inferrule*[%
    left={\textsc{#2}},%
    right={$\begin{array}{l} #5 \end{array}$}, 
    #1
  ]%
   {#3}{#4}}

   \newcommand{\Nt}{{\mathbb{M}}}

   \newcommand{\plays}[1]{\ensuremath{{\sf prt}(#1)}}

   \newcommand{\Set}[1]{\{ #1 \}}

     \newcommand{\mypath}{\sigma}

  \def\finex{{\unskip\nobreak\hfil
\penalty50\hskip1em\null\nobreak\hfil{\Large $\diamond$}
\parfillskip=0pt\finalhyphendemerits=0\endgraf}}

\newcommand{\coDefGr}{:\coDef}

\newcommand{\concat}[2]{\ensuremath{#1\,{\cdot}\,#2}}

   \newcommand{\ee}{\epsilon}

\newcommand{\RR}{\mathcal{R}}

\newcommand{\LL}[1]{{\mathcal L}(#1)}
\newcommand{\cp}[1]{{\mathsf cap}(#1)}

\newcommand{\pSet}{{\mathbf P}}
\newcommand{\pSetS}{{\mathcal P}}
\newcommand{\cml}{\Lambda}
\newcommand{\st}{\nu}
\newcommand{\eg}{\mathbb{E}^{\mathit{gl}}}

\newcommand{\Cline}[1] {\vspace{0.65mm}
\centerline{$ #1 $}\vspace{0.65mm}
}
\begin{document}

\maketitle

\begin{abstract}
MultiParty Session Types  (MPST)
provide a useful framework for safe concurrent systems.
{\em Mixed  choice 
} (enabling a participant to play 
at the same time the roles of sender and receiver) increases
the expressive power of  MPST 
as well as the difficulty in controlling safety of communications.
Such a control is more viable when modular systems are considered and the power of mixed choice fully exploited only inside loosely coupled modules. We carry over such idea in a type assignment approach to multiparty sessions. Typability for modular sessions entails Subject Reductions, Session Fidelity and Lock Freedom. 
\end{abstract}


\section{Introduction}\label{int}

MultiParty Session Types (MPST) offer a structured approach to the development and formal verification of concurrent and distributed systems~\cite{HYC08,Honda2016}. 
 As in the vast majority of choreographic formalisms,
two distinct but related views of concurrent systems are taken into account:
$(a)$ the {\em global view}, a formal specification via {\em global types} of the overall behaviour 
of a system; $(b)$ the {\em local view}, namely a description, at different  levels 
of abstraction,
of the behaviours of the single components.
A key issue in MPST, and choreographies in general, is the relation between these two views.
Among others, we can refer to  the notion of {\em projection}, used till recently in most of the MPST formalisms. 
Given a (well-formed) global type, the projection operator produces a tuple of local types -- one for each component -- which generalises  binary   session  
types~\cite{HondaK:typdyi,HVK98}.
 Such local types can be looked at as an abstraction of finer grained  descriptions of processes.  
Another approach to the MPST global-local relationship  is  
the one embodied in the so called {\em Simple MultiParty Sessions} (SMPS) formalisms.
Such an approach was first introduced in~\cite{DGD22} and~\cite{BDL22}, 
and further investigated in a bunch of papers, among  which~\cite{BD23,BDGY23,BDL23,BDL24,BBD25a,BBD25}. 
Whereas  the  MPST approach typically considers 
two-layered local views 
-- a layer of processes and a layer of local types --
SMPS are based on   
 single-layered local views, 
where only a fairly abstract notion of process is considered.
In SMPS, which is the general setting of the present paper, systems of communicating processes are represented as {\em multiparty sessions}, i.e. parallel compositions
of named processes (the {\em participants}). Then, by means of type systems, global types are
inferred for such sessions. Typability is such to ensure relevant communication properties 
-- typically Lock Freedom -- for sessions.
Besides the above mentioned ones, it is worth recalling that also other approaches have been investigated,
like the one introduced in~\cite{ScalasY19}, where the global view is only implicitly considered.

A common feature of all the above mentioned approaches, till lately, has been the use of 
communication models where,  before any interaction, a process can clearly be identified
as a sender or a receiver.  
The intrinsic potentiality of nondeterministically choosing among both inputs and outputs inside a
single process interaction (usually referred to in the literature as ``mixed choice'') 
has however recently intrigued session type researchers, both for the binary and the multiparty  cases~\cite{CMV22,PBK23,PBMK24,PY24b,PY24}. 
A thorough investigation of the expressivity of mixed  choice 
in (synchronous) MPST formalisms has been carried on in~\cite{PY24}.
For instance, 
mixed choice enables to implement protocols safely exploiting circular interactions,
as shown in~\cite{PY24} through an example recalled in the present paper (see Example~\ref{ex:Yo}).
This is not possible in usual MPST formalisms where 
typing~\cite{BDGY23} or, equivalently, the projectability condition on global types -- as shown in~\cite{BDL24} -- consists essentially in 
checking the  possibility of sequentialising the interactions in a protocol.
Together with its 
 expressive power, however, mixed  choice 
 brings subtly harmful features, as already exposed decades ago~\cite{GMY80} 
in the setting of Communicating Finite State Machines~\cite{bz83} 
 (an asynchronous formalism closely
related to MPST).
In the present paper we aim at exploiting such expressive power 
in a safe and controlled way using a SMPS setting. 
In order to do that we resort to the notion of {\em modularity}.

Modularity is a fairly general property of complex systems. Any complex system can be decomposed into smaller subsystems that are always going to be interdependent to some extent and independent to some other extent~\cite{Simon1991}.
In fact, in many human activities, from business to biology, as well as to software engineering,
modularisation offers a strategic approach enabling to cope with their complexity.
Modularity in software engineering refers to the design approach that emphasises the separation of  concerns: 
a complex software system is decomposed 
into smaller, loosely coupled modules,
where coupling is the degree of interdependence between the modules.
By means of project modularisation one manages to, among others,
reduce complexity (breaking down a large system into smaller modules makes it more manageable and easier to understand~\cite{ottia}) as well as to improve testing and separation of concerns (SoC), a fundamental principle in software engineering. 
In particular, in modular programming, concerns are separated such that modules, performing logically coherent tasks, do interact through simple and manageable interfaces.

Our proposal is hence to restrict our attention to sessions corresponding to modularised systems.  
A type discipline is then proposed that profits,
as in more rigid MPST formalisms, from a form of ``sequentialisation'' condition.
Such a condition however, instead of being imposed on participants, is imposed on the modules
forming a session, inside which the mixed  choice 
can be freely used (at the cost of a thoroughly check, but limited inside the single modules,
of all the possible interactions among participants). The inter-modules interactions 
are instead  more  controlled,  
so respecting the decoupling of modules 
characterising any sound decomposition of systems.
It is then possible to prove the properties of  Subject Reduction and Session Fidelity 
for typable modular sessions. Moreover, typability also entails the communication
property of Lock Freedom. Typability is shown to be independent from the way a session can be 
modularised as well as from the order in which the typing of the modules is ``sequentialised''.
We propose, as working example, a modular extension of the above mentioned leader election example of~\cite{PY24}.

{\em Overview.} In Section~\ref{tc} we introduce the calculus of multiparty sessions with mixed  choice 
as an extension of the SMPS calculus of~\cite{BDL22} and~\cite{BDGY23}.
Modularisable multiparty sessions are then formally presented in Section~\ref{sec:mms}.
Global types equipped with a coinductive LTS are defined in Section~\ref{sec:ts} in the style of~\cite{BDL24}. Properties of typable modular sessions are proven in Section~\ref{spr},
namely  Subject Reduction, Session Fidelity and Lock Freedom.  In  that  
section we also show that typability  does not depend   
on how a session is modularised or on the particular modules considered during typing. A summing-up section, also discussing related and future works, concludes
the paper.

\section{Multiparty Sessions with Mixed  Choice 
}\label{tc}

We  present now a SMPS synchronous calculus of multiparty sessions with  mixed choice, 
inspired  mainly 
by~\cite{PY24} and partially by~\cite{BDL22}. 
We assume to have the following denumerable base sets: \emph{messages} (ranged
over by $\msg{\lambda},\msg{\lambda'},\dots$); \emph{session participants} (ranged over
by  $\pp,\q,\pr, \ps, \ldots$); \emph{indexes} (ranged over by  $i, j, h, k,\dots$);
\emph{ finite sets of indexes} (ranged over by $I, J, H, K, \dots$). 
We refer to the denumerable set of participant names as $\mathfrak{P}$.

Processes, ranged over by $\PP,\Q,\PR,\PS,\dots$, 
implement the behaviour of participants.
In the following and in later definitions the symbol $\coDef$ does  express  
that the
productions have to be interpreted \emph{coinductively}  and that only \emph{regular} terms are allowed.  
   Then we can adopt in proofs the coinduction style
advocated in \cite{KozenS17} which, without any loss of formal rigour,
 promotes readability and conciseness.

 \begin{definition}[Processes]\label{p} 
 \begin{enumerate}[i)]
 \item
 {\em  Action  prefixes} are defined by\quad $\pi\ ::=\ \pp?\msg{\lambda}\ \mid\  \pp!\msg{\lambda}$. 
  \item
  {\em Processes} are  coinductively  defined by
  
  \Cline{\PP\coDef\inact\ \mid\ \Sigma_{i\in I}\pi_i.\PP_i}
  
  \noindent
where $I\neq\emptyset$ and finite, and 
  $\pi_{ l } 
  = \q? \msg{\lambda}_{ l }
  $, $\pi_j= \q? \msg{\lambda}_j$ (resp. $\pi_{ l }
  = \q! \msg{\lambda}_{ l }
  $, $\pi_j= \q! \msg{\lambda}_j$) imply 
 $\msg{\lambda}_l\neq\msg{\lambda}_j$, for any $ l,j  
 \in I$ such that $ l\neq j $. 
 \end{enumerate}
\end{definition}
 In the above definition, $\Sigma_{i\in I}\pi_i.\PP_i $ stands, as usual, for the summand of
processes $\pi_i.\PP_i$'s.
A $\Sigma_{i\in I}\pi_i.\PP_i $ process represents  the 
nondeterministic choice of one of the actions $\pi_i$, after which the process
continues as $\PP_i$ with $i\in I$. As usual, we assume the summand of processes to be commutative
and associative. 
A prefix $\pi$ can be any input (i.e. of the form $\pp?\msg{\lambda}$)
or output (i.e. of the form $\pp!\msg{\lambda}$) action.  
We use $\inact$ to denote the terminated process. 
For the sake of readability, we omit trailing $\inact$'s in processes.

 We  define  the participants of action prefixes by  $\plays{\pp?\msg{\lambda}} = \plays{\pp!\msg{\lambda}} = \Set{\pp}$.
 Moreover, we define the participants of processes by
 
 \Cline{ \plays{\inact}=\emptyset \qquad\plays{\Sigma_{i\in I}\pi_i.\PP_i }=\bigcup_{i\in I}\plays{\pi_i}\cup\bigcup_{i\in I}\plays{\PP_i}
          }
          
\noindent
By the regularity condition  and the finiteness of indexes,  $\plays{\PP}$ is finite for any $\PP$. 

Processes  correspond to 
 inductively defined terms of the calculus MCMP (Mixed Choice Multiparty Sessions) 
 as defined in~\cite{PY24}, where the $\mu$-operator is used to describe infinite behaviours. 
 Our use of coinductively defined (possibly) infinite terms 
 enables us to get simpler formalisations and proofs with respect to the use of $\mu$-terms.
The latter entail just technicalities that can be dealt with as done in the literature
on session types and MPST~\cite{Honda2016}, where such terms are usually considered.

\smallskip
 Multiparty sessions  are parallel compositions of  participant-process pairs, where all participants are different. 
 
\begin{definition}[Multiparty sessions] 
{\em  Multiparty sessions} are defined by\quad
$\Nt = \pP{\pp_1}{\PP_1} \parN \cdots \parN \pP{\pp_n}{\PP_n}$\\
where $\pp_j \neq \pp_l $ for $1\leq j,l\leq n$ and $j\neq l$.
\end{definition}
\noindent 
 We assume the standard structural congruence  $\equiv$  on
multiparty sessions, stating that
parallel composition is commutative and associative and has neutral
elements $\pP\pp\inact$ for any  fresh $\pp$.
Such a congruence can be inductively defined, as multiparty sessions are. 
We then 
write $\pP\pp\PP\in\Nt$ if 
$\Nt\equiv\pP\pp\PP\parN\Nt'$ and $\PP\neq\inact$. 
Moreover, we define the participants of multiparty sessions by $\plays\Nt=\set{\pp\mid \pP\pp\PP\in\Nt}$.

\smallskip
To define the {\em  synchronous operational semantics} of sessions we use an LTS,
whose transitions are decorated by {\em communication labels}, i.e. expressions of the shape 
$\pp\msg{\lambda}\q$.  In the following we use $\cml$ to range over communication labels.\\
 {\bf Notation}:  We use 
$\pi.\PP \pplus \Q$   as short for either $\pi.\PP + \Q$ or $\pi.\PP$.
 Such a notation enables us to present the following LTS in a compact and yet formal way, since
in our processes -- as in~\cite{PY24} -- we cannot have unprefixed $\inact$'s as summands.

\begin{definition}[LTS for Multiparty Sessions]\label{slts}
The {\em   labelled transition   system   (LTS) 
for multiparty sessions}  
  is the closure under structural congruence of the reduction specified by the unique axiom:

  \Cline{\begin{array}[c]{@{}c@{}}
      \NamedRule{\rn{Comm}}{  
      }{\quad
      \pP\pp  {\q!\msg{\lambda}.\PP \ppplus  \PP'}
            \parN
      \pP \q {\pp?\msg{\lambda}.\Q  \ppplus \Q'}\parN \Nt\\
     \SLTS{\pp \msg{\lambda} \q}\qquad
      \pP\pp{\PP}\parN\pP\q{\Q}\parN\Nt
      \quad
      }{}
    \end{array}
 }
 
\end{definition}
\noindent
Rule \rn{Comm} makes the communication possible: if participant
$\pp$ is enabled to send message $\msg{\lambda}$ to participant $\q$
which, in turn, is enabled to receive it, the message can be exchanged.
This rule is non-deterministic in the choice of messages exchanged.
 The implementation issues raised by such operational semantics are similar
to those for most of the calculi for concurrency and can be dealt with by resorting to suitable
coordination protocols. Such issues are however outside the scope of the present paper. 

 Note that in the above  
semantics there is no difference between the behaviours of inputs and outputs, while usually a sender freely chooses among all its available 
messages.
 In actual communicating systems, messages would also carry values that are abstracted
away in SMPS formalisms for the sake of simplicity. 
The present calculus (as well as its type system) could however be extended to messages with data.

We define {\em traces} as (possibly infinite) sequences of communication labels. Formally, 

\Cline{
\mypath\coDefGr\ee\ \mid\ \concat\cml\mypath}

\noindent
 where 
 $\ee$ is the empty sequence. 
 When $\mypath=\concat{\cml_1}{\concat\ldots{\cml_n}}$ ($n\geq 0)$
we write $\Nt\SLTS{\mypath}\Nt'$ as short for

\Cline{
\Nt\SLTS{\cml_1}\Nt_1\cdots\SLTS{\cml_n}\Nt_{n}  =  \Nt'
}

\noindent
We write $\Nt \SLTS{\sigma}$ 
with the standard meaning. 
 Moreover,  we define  the participants of labels
  by $\plays{\pp\msg{\lambda}\q} = \Set{\pp, \q}$ and 
  the participants of  traces
  -- notation   $\plays{\mypath}$ -- as its obvious extension.   
We also  denote by $\LL\Nt$ 
 the {\em set of all labels the session $\Nt$ can emit,} i.e.  $\LL\Nt=\set{\cml\mid\Nt\SLTS\cml}$.  
 
\medskip
We present now, in our setting,  the leader-election example used in~\cite{PY24} (inspired by \cite{Palamidessi03}, in turn inspired by \cite{Bouge88}) to make evident the
expressive power of mixed choice. 

\begin{example}[Leader election~\cite{PY24}]\label{ex:Yo}
\em
Five participants ($\pa,\pb,\pc,\pd,\pe$) interact with the aim of electing a leader.
Each of them can send to the next participant, in a circular fashion, the message $\msg{leader}$
in order to ask it to become the leader. If such  a  communication succeeds, the sender terminates.
Of course only two of this sort of communications can succeed.
The protocol then allows two, among the remaining three participants, to be able to 
exchange the  
$\msg{leader}$ message. The receiver is hence  
considered  as  
the $\msg{elect}$ed leader
and so it informs the station participant $\pS$ which, in turn, provides to $\msg{del}$ete the participant  which  
remained inactive during the previous interactions.  
\\[-8mm]
\end{example}
\begin{wrapfigure}{l}{0.22\textwidth}
\vspace{-12mm}
\begin{tikzpicture}[node distance=1cm]
        \node (pa) [draw,rectangle] {$\pa$};
        \node (pb) [draw,rectangle, below right of = pa,xshift = 7pt,yshift = 1pt] {$\pb$};
        \node (pc) [draw,rectangle, below left of = pb,xshift = 10pt,yshift = -10pt] {$\pc$};
        \node (pe) [draw,rectangle, below left of = pa,xshift = -7pt,yshift = 1pt] {$\pe$};       
        \node (pd) [draw,rectangle,  below right of = pe,xshift = -10pt,yshift = -10pt] {$\pd$};
        \node (pS) [draw,rectangle, right of = pb, xshift = 0pt,yshift = 0pt] {$\pS$};
        \path  (pa) edge [-stealth,bend right=20,blue] node {} (pe)
                  (pe) edge [-stealth,bend right=20,blue] node {} (pd)
                  (pd) edge [-stealth,bend right=20,blue] node {} (pc)
                  (pc) edge [-stealth,bend right=20,blue] node {} (pb)
                  (pb) edge [-stealth,bend right=20,blue] node {} (pa)
                  (pa) edge [-stealth,red] node {} (pc)
                  (pd) edge [-stealth,red] node {} (pa)
                  (pe) edge [-stealth,red] node {} (pb)
                  (pc) edge [-stealth,red] node {} (pe)
                  (pb) edge [-stealth,red] node {} (pd)
                  ;
 \end{tikzpicture}
 \vspace{-12mm}
\end{wrapfigure} 
The above behaviour is implemented by the following session:

\Cline{
\mathbb{E} = \pP\pa{\PP_{\pa}} \parN  \pP\pb{\PP_{\pb}} \parN \pP\pc{\PP_{\pc}} \parN\pP\pd{\PP_{\pd}} \parN  \pP\pe{\PP_{\pe}}\parN\ \pP{\pS}{\PP_{\pS}}
         }
  
  \noindent       
$
\begin{array}{llcl}
\text{where} & \PP_{\pa} & = & \ \ \, \pe!\msg{leader}\\
             &  &    & + \pb?\msg{leader}.(\pc!\msg{leader} + \pd?\msg{leader}.\pS!\msg{elect})\\
             &  &    &   + \pS?\msg{del} 
\end{array}
$

and (with some abuse of notation) \\
$
\begin{array}{l@{\hspace{37mm}}lcl}
 \qquad\qquad& \PP_{\pS} & = &  \Sigma_{\px\in\Set{\pa,\pb,\pc,\pd,\pe}}\big(\px?\msg{elect}.\Sigma_{\px\in\Set{\pa,\pb,\pc,\pd,\pe}}\px!\msg{del}\big)
\end{array}
$

\noindent
 Processes of participants $\pb$,  $\pc$, $\pd$ and $\pe$ are obtained out of $\PP_\pa$  by applying the 
name substitution \linebreak
$\st = [\pa{\mapsto}\pb,\pb{\mapsto}\pc,\pc{\mapsto}\pd,\pd{\mapsto}\pe,\pe{\mapsto}\pa]$
as follows: 

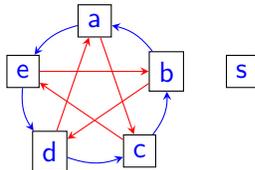
\Cline{
\PP_{\pb} = \PP_{\pa}\st,\ \PP_{\pc} = \PP_{\pb}\st,\ \PP_{\pd} = \PP_{\pc}\st,\
\PP_{\pe} = \PP_{\pd}\st}

Session $\mathbb{E}$ can be graphically represented by the diagram above on the left, where blue arrows represent the
initial possible exchanges of the message $\msg{leader}$ and the red  ones 
the potential
further exchanges of such message.
The following sequence of reductions is, for instance, the one leading to the election of $\pe$:

\Cline{\quad\qquad\hspace{32mm}\pa\msg{leader}\pe\cdot \pd\msg{leader}\pc\cdot \pc\msg{leader}\pe \cdot 
          \pe\msg{elect}\pS \cdot \pS\msg{del}\pb       \qquad\qquad\qquad\qquad\qquad\quad\text{\finex}
          }

\smallskip

Lock Freedom is a relevant property of concurrent systems.
We define it in our setting following ~\cite{Padovani14}: roughly,  
there is always a continuation enabling a participant to communicate
whenever it is willing to do so.
Lock Freedom entails Deadlock Freedom, since it ensures progress for each participant.

\begin{definition}[Lock Freedom]\label{d:lf}
  A session  $\Nt$ is \emph{lock free} if  $\Nt\SLTS{\mypath}\Nt'$ with $\mypath$ finite and
  $\pp\in\plays{\Nt'}$ imply 
	 $\Nt'\SLTS{\concat{\mypath'}\Lambda}$ 
	 for some $\mypath'$ and $\Lambda$ such that  $\pp\not\in\plays{\mypath'}$ and  $\pp\in\plays\Lambda$.
\end{definition}

The above definition corresponds to the notion of liveness used in \cite{KS10}
and \cite{LNTY17} in a channel-based synchronous communication setting.

\section{Modular Multiparty Sessions}\label{sec:mms}

``With great [expressive] power comes great responsibility'' (Spider-Man's Uncle Ben),
since expressive power is often difficult to control and tame.
Our aim is to provide a type system for multiparty sessions with mixed choice ensuring,
like usual in SMPS, relevant communication properties, together with the guarantee that
the overall behaviour of a session faithfully  respects 
what the type assigned to the session,  if any, describes.   
To do that, instead of restricting the expressive power of mixed choice, we decided
to consider multiparty sessions 
 that can be 
-- as suggested by a well-known software engineering principle -- modularised. 
A module, in our SMPS setting, is formalised in terms of a subsession inside which participants
can freely interact by means of mixed choice. 
The communications among the modules  are  instead controlled by imposing  them 
to be performed only  by particular participants called ``connectors''.  
The processes of the connectors (dubbed connecting processes) must satisfy the restriction that each choice involving a participant not belonging to the module must be between communications with only that participant. Definition~\ref{cpr}, where $\pSet$ is the set of module participants, formalises this condition.

\begin{definition}[Connecting processes]\label{cpr}
Given a set of participants $\pSet$,
 we say that a process $\PP$ is {\em $\pSet$-connecting} if for any subprocess of $\PP$, 
 say  $\Sigma_{i\in I}\pi_i.\PP_i $,  we have that $\plays{\pi_j}\not\in \pSet$ for some $j\in I$ implies $\plays{\pi_i}=\plays{\pi_j}$
 for all $i\in I$. 
 \end{definition}
 A connector is hence a participant of a module (represented by the session $\Nt$ in the  definition
 below) whose process is a connecting process which can interact with the outside of the module.
 \begin{definition}[Connectors]\label{def:conn}
Let $\pP\pp\PP\in\Nt$.
We say that the participant $\pp$ is a {\em connector}  for   
$\Nt$ if
$\PP$ is $\plays\Nt$-connecting and there is $\pq\in\plays{\PP}$ such that $\pq\not\in\plays{\Nt}$.
 \end{definition}

 The notions of subsession and session partition are at the basis of that of modular session.
 We say that $\Nt'$ is a {\em subsession} of $\Nt$ and write $\Nt'\subseteq\Nt$, whenever $\pP\pp\PP\in\Nt'$
implies  $\pP\pp\PP\in\Nt$.
A {\em partition} of a session $\Nt$ is, as expected, a set of subsessions $\Set{\Nt_h}_{h\in H}$ of $\Nt$ such that $\plays{\Nt}=\bigcup_{h\in H}\plays{{\Nt}_h}$ and, for all $h,k\in H$, $\plays{{\Nt}_h}\cap\plays{{\Nt}_k}=\emptyset$.  Therefore, 
$\pP\pp\PP\in\Nt$ implies that there is a unique $k\in H$ 
such that $\pP\pp\PP\in{\Nt}_k$.

\begin{definition}[$\mathcal{P}$-partition]
\label{def:ppartition}
Let $\Set{\Nt_h}_{h\in H}$ be a partition of a session $\Nt$, and let 
 $\mathcal{P}=\Set{\pSet_k}_{k\in K}$ be a partition of a finite subset of the set $\mathfrak{P}$. 
We say that $\Set{\Nt_h}_{h\in H}$ is a $\mathcal{P}$-{\em partition} of $\Nt$ if
$H\subseteq K$ and 
 $\plays{\Nt_h}\subseteq\pSet_h$   for all $h\in H$. 
\end{definition}
It is not difficult  
 to check that, given a session $\Nt$ and a partition $\mathcal{P}=\Set{\pSet_k}_{k\in K}$ such that $\plays\Nt\subseteq \bigcup_{k\in K}\pSet_k$, there is a unique $\mathcal{P}$-{\em partition} of $\Nt$.

\smallskip
A session  is  
modularisable  (with respect to a partition of participants)  when it can be partitioned into subsessions that interact only by means of connectors.

\begin{definition}[$\mathcal{P}$-modularisation]
\label{modularisation}
A $\mathcal{P}$-{\em modularisation of} $\Nt$ is a $\mathcal{P}$-{\em partition} 
$\Set{\Nt_h}_{h\in H}$ of $\Nt$ such that, for all $h\in H$, the following conditions hold
\begin{enumerate}[i)]
\item
\label{modularisation1}
$\pP\pp\PP\in\Nt_h$ implies  that either $\pp$ is a connector of $\Nt_h$ or, 
 for each $\pq\in \plays{\PP}$, $\pq\in\plays{\Nt}$ implies $\pq\in \plays{\Nt_h}$; 
\item
\label{modularisation2}
 for each connector $\pP\pp\PP\in\Nt_h$ and each $\pq\in\plays{\PP}\setminus\plays{\Nt_h}$:
 
\Cline{\text{if }\pq\in\plays{\Nt_k}, 
\text{ then  }\pq\text{ is a connector for }\Nt_k}
\end{enumerate}
In such a case, we say that $\Nt$ is $\mathcal{P}$-modularisable.
\end{definition}

It is worth noticing that we impose no limit on the number of connectors present in a
module, as well as on the number of external connectors a connector can interact with
 (see the following example). 
 Besides, a session $\Nt$ is always $\Set{\plays{\Nt}}$-modularisable.
 For example  $\Nt=\pP\pp{\pq!\msg{\lambda}+\pr!\msg{\lambda'}}$ is $\Set{\Set{\pp}}$-modularisable  and its unique module does not contain any connector.
Also, given a partition $\mathcal{P}$ and a session $\Nt$, there exists a unique $\mathcal{P}$-modularisation of  $\Nt$, if any. 
\begin{example}[Modules with multiple connectors]\label{ex:mconn}
{\em
 Let us consider
$\Nt = 
\Nt_1 \parN \Nt_2$,
\\ 
where
$\Nt_1 =
\pP\pu{\pp?\msg{\lambda}.\pq!\msg{\lambda}\ +\ \pq!\msg{\lambda}.\pp?\msg{\lambda}}
\parN
\pP\pp{\pu!\msg{\lambda}.(\pr!\msg{\lambda}_1+\pr?\msg{\lambda}_2)} 
\parN 
\pP\pq{\pu?\msg{\lambda}.(\ps!\msg{\lambda}_1+\ps?\msg{\lambda}_2)}
$\\ 
and
$\quad\Nt_2 = 
\pP\pv{\pr!\msg{\lambda}.\ps?\msg{\lambda}\ +\ \ps?\msg{\lambda}.\pr!\msg{\lambda}}
\parN
\pP\pr{\pv?\msg{\lambda}.(\pp?\msg{\lambda}_1+\pp!\msg{\lambda}_2)} 
\parN 
\pP\ps{\pv!\msg{\lambda}.(\pq?\msg{\lambda}_1+\pq!\msg{\lambda}_2)}
$. 
\\ 
This session is $\Set{\Set{\pu,\pp,\pq},\Set{\pv,\pr,\ps}}$-modularisable and it 
has multiple connectors.  
In fact, $\pp$ and $\pq$ are the connectors for the module $\Set{\pu,\pp,\pq}$,
whereas $\pr$ and $\ps$ are the connectors for the module $\Set{\pv,\pr,\ps}$. 
\finex
}
\end{example}

\medskip


 In actual programming, refining a modularisation enables to enhance parameters like 
scalability, maintenance, reusability and many more.
In the present setting, it also allows for simpler typings. 
Any modularisation refinement corresponds to a partition refinement. 

\begin{definition}[Refinement] A partition $\pSetS$ refines a partition $\pSetS'$ (notation $\pSetS \sqsubseteq\pSetS'$) if $\pSetS=\set{\pSet_h}_{h\in H}$, $\pSetS'=\set{\pSet'_k}_{k\in K}$ and for all $k\in K$ there is $H_k\subseteq H$ such that $\pSet_k'=\bigcup_{h\in H_k}\pSet_h$.
\end{definition}

As intuitively evident, it is possible to formally show that coarser partitions maintain 
modularisability.

\begin{lemma}
If $\Nt$ is $\pSetS$-modularisable and $\pSetS$ refines $\pSetS'$, then $\Nt$ is $\pSetS'$-modularisable.
\end{lemma}
\begin{proof} It is  not difficult  
to verify that if conditions (\ref{modularisation1}) and (\ref{modularisation2}) of Definition~\ref{modularisation} hold for $\pSetS$, then they hold also for $\pSetS'$. \qed
\end{proof}

\noindent
From  the  previous lemma,  being $\Nt$ $\Set{\plays{\Nt}}$-modularisable, it immediately follows that 
$\Nt$ is $\pSetS$-modularisable for all $\pSetS$ such that $\plays\Nt\subseteq\pSet$ for some $\pSet\in\pSetS$.

 We can build  a minimal  
refined partition, with respect to $\sqsubseteq$,  among those  modularising a session. 
We start with $\pSetS_0=\set{\set\pp\mid \pp\in\plays\Nt}$
 and we iteratively build $\pSetS_{i+1}$ by replacing in $\pSetS_i$ the two sets $\pSet,\pSet'$ with the unique set $\pSet\cup\pSet'$ if there is $\pP\pp\PP\in\Nt$ such that $\pp\in\pSet$, $\pq\in\pSet'\cap\plays\PP$ and either $\PP$ is not $\pSet$-connecting or $\pq$ is not a connector for the subsession of $\Nt$ whose set of participants is $\pSet'$.  It is  possible 
 to  verify 
 that $\Nt$ is $\pSetS$-modularisable, where $\pSetS$ is the fixed point of this procedure. 
 We show now such a partition to be also the minimum among the partitions modularising 
 a session.
\begin{lemma}
The  minimal   
partition 
modularising a session is unique,  i.e. a minimum.  
\end{lemma}
\begin{proof}
Assume toward a contradiction that there are two different  minimal (w.r.t $\sqsubseteq$) 
partitions $\pSetS=\set{\pSet_h}_{h\in H}$ and $\pSetS'=\set{\pSet'_k}_{k\in K}$ for modularising a session. This implies that there are $h_1, h_2\in H$ and $k_0\in K$ such that $\pp\in\pSet_{h_1}$,  $\pq\in\pSet_{h_2}$ and $\set{\pp,\pq}\subseteq\pSet'_{k_0}$. Then also the partition obtained from $\pSetS'$ by replacing the set $\pSet'_{k_0}$ with the two sets $\pSet'_{k_0}\cap\pSet_{h_1}$ and 
$\pSet'_{k_0}\cap\pSet_{h_2}$ modularises the same session. This is clearly a contradiction. \\[-6mm]
\begin{flushright}\qed\end{flushright}
\end{proof}

 It is crucial that $\pSetS$-modularisation is preserved by reduction. 

\begin{lemma}
Let $\Nt$ be $\pSetS$-modularisable and let $\Nt\SLTS{\cml}\Nt'$. 
Then also $\Nt'$ is $\mathcal{P}$-modularisable.
\end{lemma}
\begin{proof} Notice that conditions (\ref{modularisation1}) and (\ref{modularisation2}) of Definition~\ref{modularisation} are invariant by reduction. In fact a participant $\pp$ which is a connector in $\Nt$ is either a connector in $\Nt'$ too or it has a process whose participants all belong to the subsession containing $\pp$ in the $\mathcal{P}$-modulation of $\Nt'$.
Therefore a $\mathcal{P}$-modularisation of $\Nt$ is also a $\mathcal{P}$-modularisation of $\Nt'$. 
\qed
\end{proof}

\smallskip
 We can notice that,  
if $\pSetS$ is the  minimal
partition for modularising $\Nt$ and $\Nt\SLTS{\cml}\Nt'$, in general $\pSetS$ is not the  minimal
partition for modularising $\Nt'$. In fact $\plays{\Nt'}$ can be a proper subset of $\plays{\Nt}$ and a process in $\Nt$ which is not a connector can reduce to a process which is a connector in $\Nt'$. 
For example the minimal
partition of $\Nt\equiv \pP\pp{\pq!\msg{\lambda}.\pr!\msg{\lambda} + \pr?\msg{\lambda'}.(\pq!\msg{\lambda}_1+\pq?\msg{\lambda}_2)}\parN\pP\pq{\pp?\msg{\lambda}+\pp?\msg{\lambda}_1+\pp!\msg{\lambda}_2}\parN\pP\pr{\pp?\msg{\lambda}+\pp!\msg{\lambda'}}$ is $\set{\set{\pp,\pq,\pr}}$, since $\pp$ is not a connector. But $\Nt
\SLTS{\pr\msg{\lambda'}\pp} \pP\pp{\pq!\msg{\lambda}_1+\pq?\msg{\lambda}_2}\parN\pP\pq{\pp?\msg{\lambda}+\pp?\msg{\lambda}_1+\pp!\msg{\lambda}_2}$ and the  minimal
partition of this last session is $\set{\set\pp,\set\pq}$, since $\pp$ and $\pq$ are connectors and $\pr$ disappeared. 

\begin{example}[Modular election]\label{ex:gelection}
{\em
We consider three ``local'' elections, all managed like in the Example \ref{ex:Yo}.
The names of the participants of the three local election are like the ones
in the Example \ref{ex:Yo}, but indexed with indexes in $\Set{1,2,3}$.
We also consider a further ``global election'' with participants $\pw_1$, $\pw_2$ and $\pw_3$ and  global  station  
$\pgS$. 
Such an election follows a protocol similar to that of the local elections (but simpler, since only three
participants do compete for leadership).
It can be seen as an election among the 
$\ptp{w}$inners of the local elections.
In the present example  the ``local leaders'', once they are elected, are informed whether they have been elected also ``global leader'' or not.
The above sketched global behaviour is implemented by the following session
 $\eg$ 
made of four subsessions: three local elections 
 ($ {\mathbb{E}}_1,  {\mathbb{E}}_2, {\mathbb{E}}_3$) 
and 
  one global election   ${\mathbb{G}}$ among the ``local'' $\pw$inners.

The processes of the participants,
  are fairly similar to the ones in the Example \ref{ex:Yo} to which a part is added implementing
  the communication to the local leaders of whether they are also global leader or not.

\begin{figure}[h!]
\begin{center}
$
\begin{array}{c@{\hspace{-8mm}}c@{\hspace{-8mm}}c}
&
\scalebox{1}{
\begin{tikzpicture}[node distance=1cm]
        \node (pa1) [draw,rectangle] {$\pa_1$};
        \node (pb1) [draw,rectangle, below right of = pa1,xshift = 7pt,yshift = 1pt] {$\pb_1$};
        \node (pc1) [draw,rectangle, below left of = pb1,xshift = 10pt,yshift = -10pt] {$\pc_1$};
        \node (pe1) [draw,rectangle, below left of = pa1,xshift = -7pt,yshift = 1pt] {$\pe_1$};       
        \node (pd1) [draw,rectangle,  below right of = pe1,xshift = -10pt,yshift = -10pt] {$\pd_1$};
        \node (pS1) [draw,rectangle, below of = pa1, xshift = 0pt,yshift = -46pt] {$\ptp{\pS}_1$};
        \path  (pa1) edge [-stealth,bend right=20,blue] node {} (pe1)
                  (pe1) edge [-stealth,bend right=20,blue] node {} (pd1)
                  (pd1) edge [-stealth,bend right=20,blue] node {} (pc1)
                  (pc1) edge [-stealth,bend right=20,blue] node {} (pb1)
                  (pb1) edge [-stealth,bend right=20,blue] node {} (pa1)
                  (pa1) edge [-stealth,red] node {} (pc1)
                  (pd1) edge [-stealth,red] node {} (pa1)
                  (pe1) edge [-stealth,red] node {} (pb1)
                  (pc1) edge [-stealth,red] node {} (pe1)
                  (pb1) edge [-stealth,red] node {} (pd1)
                  ;
 \end{tikzpicture}
 }
 \\[-5mm]
\scalebox{1}{
\begin{tikzpicture}[node distance=1cm]
        \node (pa3) [draw,rectangle] {$\pa_3$};
        \node (pb3) [draw,rectangle, below right of = pa3,xshift = 7pt,yshift = 1pt] {$\pb_3$};
        \node (pc3) [draw,rectangle, below left of = pb3,xshift = 10pt,yshift = -10pt] {$\pc_3$};
        \node (pe3) [draw,rectangle, below left of = pa3,xshift = -7pt,yshift = 1pt] {$\pe_3$};       
        \node (pd3) [draw,rectangle,  below right of = pe3,xshift = -10pt,yshift = -10pt] {$\pd_3$};
        \node (pS3) [draw,rectangle, right of = pb3, xshift = 0pt,yshift = 0pt] {$\ptp{\pS}_3$};
        \path  (pa3) edge [-stealth,bend right=20,blue] node {} (pe3)
                  (pe3) edge [-stealth,bend right=20,blue] node {} (pd3)
                  (pd3) edge [-stealth,bend right=20,blue] node {} (pc3)
                  (pc3) edge [-stealth,bend right=20,blue] node {} (pb3)
                  (pb3) edge [-stealth,bend right=20,blue] node {} (pa3)
                  (pa3) edge [-stealth,red] node {} (pc3)
                  (pd3) edge [-stealth,red] node {} (pa3)
                  (pe3) edge [-stealth,red] node {} (pb3)
                  (pc3) edge [-stealth,red] node {} (pe3)
                  (pb3) edge [-stealth,red] node {} (pd3)
                  ;
 \end{tikzpicture}
 }
 &
 &
 \scalebox{1}{
\begin{tikzpicture}[node distance=1cm]
        \node (pa2) [draw,rectangle] {$\pa_2$};
        \node (pb2) [draw,rectangle, below right of = pa2,xshift = 7pt,yshift = 1pt] {$\pb_2$};
        \node (pc2) [draw,rectangle, below left of = pb2,xshift = 10pt,yshift = -10pt] {$\pc_2$};
        \node (pe2) [draw,rectangle, below left of = pa2,xshift = -7pt,yshift = 1pt] {$\pe_2$};       
        \node (pd2) [draw,rectangle,  below right of = pe2,xshift = -10pt,yshift = -10pt] {$\pd_2$};
        \node (pS2) [draw,rectangle, left of = pe2, xshift = 0pt,yshift = 0pt] {$\ptp{\pS}_2$};
        \path  (pa2) edge [-stealth,bend right=20,blue] node {} (pe2)
                  (pe2) edge [-stealth,bend right=20,blue] node {} (pd2)
                  (pd2) edge [-stealth,bend right=20,blue] node {} (pc2)
                  (pc2) edge [-stealth,bend right=20,blue] node {} (pb2)
                  (pb2) edge [-stealth,bend right=20,blue] node {} (pa2)
                  (pa2) edge [-stealth,red] node {} (pc2)
                  (pd2) edge [-stealth,red] node {} (pa2)
                  (pe2) edge [-stealth,red] node {} (pb2)
                  (pc2) edge [-stealth,red] node {} (pe2)
                  (pb2) edge [-stealth,red] node {} (pd2)
                  ;
 \end{tikzpicture}
 } 
\\[-8mm]
&
\scalebox{1}{
\begin{tikzpicture}[node distance=1cm]
        \node (pw1) [draw,rectangle] {$\pw_1$};
        \node (pw2) [draw,rectangle, below right of = pw1,xshift = 0pt,yshift = -7pt] {$\pw_2$};
        \node (pw3) [draw,rectangle, below left of = pw1,xshift = 0pt,yshift = -7pt] {$\pw_3$};
        \node (gS) [draw,rectangle, above of = pw1,xshift = 0pt,yshift =0pt] {$\pgS$};
        \path  (pw3) edge [-stealth,bend right=35,blue] node {} (pw2)
                  (pw2) edge [-stealth,bend right=35,blue] node {} (pw1)
                  (pw1) edge [-stealth,bend right=35,blue]  node {} (pw3)
                  ;
 \end{tikzpicture}
 }
\end{array}
 $
 \end{center}
\caption{The three local elections and the global one of Example \ref{ex:gelection}.}\label{fig:glsession} 
\end{figure}
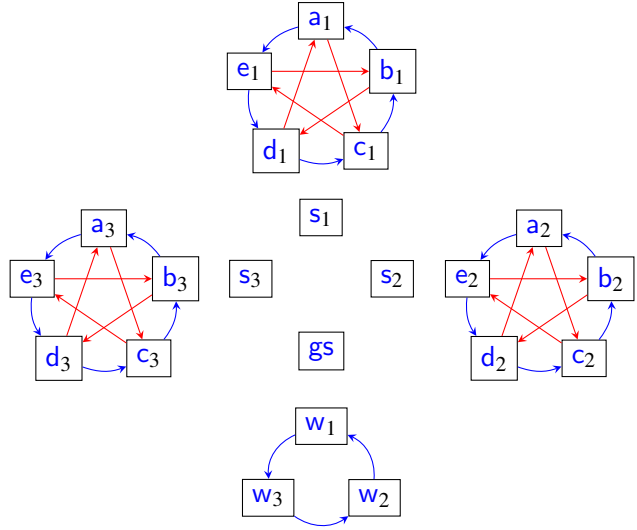

 \Cline{
\begin{array}{lcl}
\eg & = & {\mathbb{E}}_1 \parN {\mathbb{E}}_2 \parN  {\mathbb{E}}_3 \parN {\mathbb{G}}
\end{array}
}


\noindent 
 where, for  $1\leq i\leq 3$, 
 
  \Cline{
\begin{array}{lcl}
{\mathbb{E}}_i  
& = & \pP{\pa_i}{\PP_{\pa_i}} \parN  \pP{\pb_i}{\PP_{\pb_i}} \parN \pP{\pc_i}{\PP_{\pc_i}} \parN \pP{\pd_i}{\PP_{\pd_i}} \parN  \pP{\pe_i}{\PP_{\pe_i}}\ \parN\ \pP{\pS_i}{\PP_{\pS_i}}
\end{array}
}

\noindent
$
\begin{array}{llcl}
 \qquad
 \text{with} & \PP_{\pa_i} & = & \ \ \ \ \pe_i!\msg{leader}\\
             &  &    & +\ {\pb_i}?\msg{leader}.\Big(\pc_i!\msg{leader} + \pd_i?\msg{leader}.\pS_i!\msg{elect}.\big(\pS_i?\msg{gleader}\ +\ \pS_i?\msg{no}\big)\Big)\\
             &  &    &   +\ \pS_i?\msg{del} 
\end{array}
$

\noindent
$
\begin{array}{llcl}
\qquad
\text{and} & \PP_{\pS_i} & = &  \Sigma_{\px\in\Set{\pa_i,\pb_i,\pc_i,\pd_i,\pe_i}}\px?\msg{elect}.\big(\pgS?\msg{gleader}.\px!\msg{gleader}.\D_i\ +\ \pgS?\msg{no}.\px!\msg{no}.\D_i\big)\\[1mm]
  & & & \text{with }  \D_i  =\Sigma_{\px\in\Set{\pa_i,\pb_i,\pc_i,\pd_i,\pe_i}}\px!\msg{del}\\
& \multicolumn{3}{l}
{\text{and with  processes of participants $\pb_i$,  $\pc_i$, $\pd_i$ and $\pe_i$ obtained out of $\PP_{\pa_i}$  by}}\\
& \multicolumn{3}{l}
{\text{applying the 
name substitution}}\\
& \multicolumn{3}{l}
{\qquad\qquad\text{$\st_i 
= [\pa_i{\mapsto}\pb_i,\pb_i{\mapsto}\pc_i,\pc_i{\mapsto}\pd_i,\pd_i{\mapsto}\pe_i,\pe_i{\mapsto}\pa_i]$
 as in Example \ref{ex:Yo}}
}
\end{array}
$

\noindent
\\
\noindent
and 

 \Cline{
\begin{array}{lcl}
 \mathbb{G} & = & \pP{\pw_1}{\PP_{\pw_1}} \parN  \pP{\pw_2}{\PP_{\pw_2}} \parN \pP{\pw_3}{\PP_{\pw3}}\ \parN\ \pP{\pgS}{\PP_{\pgS}}
\end{array}
}

\noindent
$
\begin{array}{llcl}
\qquad\text{with} & \PP_{\pw_i} & = & \ \ \ \ \pw_{i+2}!\msg{leader}\\
              & &    & +\ {\pw_{i+1}}?\msg{leader}.\pgS!\msg{gleader}\\
              & &    &   +\ \pgS?\msg{del} 
\end{array}
$

\noindent
$
\begin{array}{llcl}
\qquad \text{and} & \PP_{\pgS} & = &  \Sigma_{i\in\Set{1,2,3}}\pw_i?\msg{gleader}.\pS_i!\msg{gleader}.\pS_{i+1}!\msg{no}.\pS_{i+2}!\msg{no}.\big(  \Sigma_{i\in\Set{1,2,3}}\pw_i!\msg{del}\big)\\
 & \multicolumn{3}{l}
{\text{where all the indexes above have to be considered modulo 3,  plus 1.}}
\end{array}
$
\medskip

\noindent
It is not difficult to check that $\eg$ is $\pSetS$-modularisable for
$\pSetS=\Set{\plays{{\mathbb{E}}_1},\plays{{\mathbb{E}}_2},\plays{{{\mathbb{E}}}_3},\plays{{\mathbb{G}}}}$,
where $\pS_1,\pS_2,\pS_3$ and $\pgS$ are the connectors
for, respectively, the modules 
 ${\mathbb{E}}_1$, ${\mathbb{E}}_2$, ${\mathbb{E}}_3$ and
 ${\mathbb{G}}$. 
}
\finex
\end{example}

\section{A Type System for Modular Sessions}\label{sec:ts}

Global types are used to represent the overall behaviour of multiparty sessions.
Here we use a notion of global type similar to the one in  MPST but, following SMPS,
we define them coinductively, as possibly infinite regular terms.

\begin{definition}[Global types]\label{def:gt}
{\em Global types} are  coinductively  defined by:

\Cline{\G\coDef\End\ \mid\ \Sigma_{i\in I}\cml_i.\G_i
         }

\noindent
where $I\neq\emptyset$ and finite, and for any $j, l\in I$ such that $j\neq l$, 
$\cml_i= \pp\msg{\lambda}_i\q$ and $\cml_j= \pp\msg{\lambda}_j\q$ imply $\msg{\lambda}_j\neq\msg{\lambda}_l$.  
\end{definition}

We define  the participants of global  types   
 by $\plays{\End}=\emptyset$ and $\plays{\Sigma_{i\in I}\cml_i.\G_i}=\bigcup_{i\in I}\plays{\cml_i}\cup\bigcup_{i\in I}\plays{\G_i}$.  By the regularity condition, $\plays{\G}$ is finite for any $\G$.
As usual, trailing $\End$'s will be omitted.

\smallskip
As mentioned in the Introduction, 
in standard SMPS, typing rules take into account single interactions
between pairs of participants.
 In our mixed-choice setting, that approach would allow to type  only 
sessions whose independent parts were 
intrinsically sequential, so ruling out protocols
like the leader election and the modular election.
 We hence consider a typing rule 
 where all the interactions involving the participants of
 a single module  $\widehat\Nt$   (so including also the communications of the connectors of 
  $\widehat\Nt$  with other modules) are taken into account. 
 Such a set of interactions is formalised through the notion of coherent set of 
 communication labels. 
 A module $\widehat\Nt$ is therefore indirectly represented in the rule in terms of 
 its corresponding coherent set 
 and referred to, when necessary, as the {\em witness} of such a set.

\begin{definition}[Coherent set of communication labels]\label{csrl}
A set of labels $\set{\cml_i}_{i\in I}$ is  $\mathcal{P}$-{\em coherent for $\Nt$}  if  $\Nt$ is $\mathcal{P}$-modularisable and 
there exists an element $\widehat\Nt$ of the (unique) $\mathcal{P}$-modularisation of $\Nt$ such that

\Cline{
\set{\cml_i}_{i\in I} = \Set{\cml\in\LL\Nt  \mid   \plays{\cml}\cap \plays{\widehat\Nt}\neq\emptyset}
} 

\noindent
The $\widehat\Nt$ above is called {\em witness} for the $\mathcal{P}$-coherence of  
$\set{\cml_i}_{i\in I}$.
\end{definition}

\begin{example}[Witness]\label{ex:awit}
{\em
Let us consider 
$\Nt'_1 =
\pP\pu{\pq!\msg{\lambda}}
\parN
\pP\pp{\pr!\msg{\lambda}_1+\pr?\msg{\lambda}_2} 
\parN 
\pP\pq{\pu?\msg{\lambda}.(\ps!\msg{\lambda}_1+\ps?\msg{\lambda}_2)}
$
and
$\Nt_2 = 
\pP\pv{\pr!\msg{\lambda}.\ps?\msg{\lambda}\ +\ \ps?\msg{\lambda}.\pr!\msg{\lambda}}
\parN
\pP\pr{\pv?\msg{\lambda}.(\pp?\msg{\lambda}_1+\pp!\msg{\lambda}_2)} 
\parN 
\pP\ps{\pv!\msg{\lambda}.(\pq?\msg{\lambda}_1+\pq!\msg{\lambda}_2)}
$. 
The session 
 $\Nt'_1 \parN \Nt_2$ can be obtained by reducing the session of 
 Example \ref{ex:mconn}. It 
is still $\Set{\Set{\pu,\pp,\pq},\Set{\pv,\pr,\ps}}$-modularisable
and   $\Nt'_1$ is the witness for the $\Set{\Set{\pu,\pp,\pq},\Set{\pv,\pr,\ps}}$-coherent 
set of labels $\Set{\pu \msg{\lambda}\pq,\, \pp\msg{\lambda_1}\pr,\, \pr\msg{\lambda_2}\pp}.$
\finex
}
\end{example}

 Here and in the following, the double line indicates that the rules are interpreted coinductively.  

\begin{definition}[Type system]\label{def:type-system}
The type system $\derP$ is defined by the following axiom and rule, where sessions are considered modulo structural congruence:  

\smallskip

\Cline{\NamedRule{\rn{End}}{}{ \End\derP\pP\pp\inact  
}{}{}}

\medskip

\Cline{
	\NamedCoRule{\rn{TComm}}
{\mbox{$\begin{array}{c} \Nt\SLTS{\cml_i}\Nt_i\qquad  {\G_i}\derP{\Nt_i}  
\qquad\forall i\in I\neq\emptyset\\
\set{\cml_i}_{i\in I}\text{ is $\mathcal{P}$-coherent for } \Nt \qquad\plays{\Sigma_{i\in I}\cml_i.\G_i}=\plays\Nt \\[3mm]
\end{array}$}}
{{\Sigma_{i\in I}\cml_i.\G_i}\derP{\Nt}}
{ }{}
	}
\end{definition}
It is not difficult to check that  we can derive the global type 
$\pu\msg{\lambda}\q.\pp\msg{\lambda}\pu.\G_1+\pp\msg{\lambda}\pu.\pu\msg{\lambda}\q.\G_1$, where $\G_1=\pv\msg{\lambda}\pr.\ps\msg{\lambda}\pv.\G_2+\ps\msg{\lambda}\pv.\pv\msg{\lambda}\pr.\G_2$, $\G_2=\pp\msg{\lambda}_1\pr.\G_3+ \pr\msg{\lambda}_2\pp.\G_3$, and $\G_3=\q\msg{\lambda}_1\ps + \ps\msg{\lambda}_2\pq$,
for the session of Example~\ref{ex:mconn}.

 The condition ``$\set{\cml_i}_{i\in I}$ is $\mathcal{P}$-coherent for $\Nt$'' is essential to get Subject Reduction. 
 In fact, by allowing any subset of $\LL\Nt $ as $\set{\cml_i}_{i\in I}$ in $\rn{TComm}$,
we could derive  $\tyn{\pp\msg{\lambda'}\pr.\pp\msg{\lambda}\pq}{\Nt_0}$ for 

\Cline{\Nt_0\equiv \pP\pp{\pq!\msg{\lambda}+\pr!\msg{\lambda'}.\pq!\msg{\lambda}}\parN\pP\pq{\pp?\msg{\lambda}}\parN\pP\pr{\pp?\msg{\lambda'}}} 

\noindent
regardless of $\mathcal{P}$.
However, we would also have $\Nt_0\SLTS{\pp\msg{\lambda}\pq} \pP\pr{\pp?\msg{\lambda'}}$,
with $\pP\pr{\pp?\msg{\lambda'}}$ untypable.
The above example also shows that, in order to get  Subject Reduction 
it is necessary that, at any moment, a connector can 
interact with one other connector only. Let us  assume 
to relax Definition \ref{cpr}
as follows
\begin{quote}\em
Given a set of participants $\pSet$,
 we say that a process $\PP$ is {\em $\pSet$-connecting} if for any subprocess of $\PP$, 
 say $\Sigma_{i\in I}\pi_i.\PP_i $, we have that $\plays{\pi_j}\not\in \pSet$ for some $j\in I$ implies $\plays{\pi_i}\not\in \pSet$ for all $i\in I$. 
\end{quote}
This would imply $\Nt_0$ above to be $\Set{\Set{\pp},\Set{\pq},\Set{\pr}}$-modularisable and
all the participants would turn out to be connectors in their respective modules. Hence
$\Nt_0$ would be $\vdash^{\Set{\Set{\pp},\Set{\pq},\Set{\pr}}}$ typable, 
whereas $\pP\pr{\pp?\msg{\lambda'}}$ would not.

  The condition ``$\plays{\Sigma_{i\in I}\cml_i.\G_i}=\plays\Nt$'' 
is  necessary 
to get Lock Freedom. For example without this condition we could derive 
${\G}\vdash^{\Set{\Set{\pp},\Set{\pq},\Set{\pr}}}{\pP\pp\PP\parN\pP\pq\PQ\parN\pP\pr{\ps!\lambda}}$ with $\PP=\pq!\msg{\lambda}.\PP$, $\PQ=\pp?\msg{\lambda}.\PQ$
and $\G=\pp\msg{\lambda}\pq.\G$.  For what concerns the condition $I\neq \emptyset$, let us consider one of our previous examples,
namely $\Nt=\pP\pp{\pq!\msg{\lambda}+\pr!\msg{\lambda'}}$. 
Such a session can be uniquely modularised with itself as possible module and it is not lock free.
In fact it is not typable, since its unique coherent set is empty. 

\smallskip

 It can be proved that   $\mathcal{P}$-coherence  
 of a label set for a session is preserved by reducing the session 
with a label not belonging to the   $\mathcal{P}$-coherent set, see Lemma~\ref{lem:coherpres}.

The type system is decidable, since processes and global types are regular, and there is only a finite number of partitions for the participants of a session.  More interesting, typability of a session does depend  on  the choice neither of the $\mathcal P$-coherent sets nor of the partition $\mathcal P$  
 (see, respectively, 
  Theorems~\ref{tpp} and~\ref{ttpp}).

\medskip
 We define the semantics of global types via a coinductive formal system, as done  first 
 in~\cite{BDL24}. 
Such a coinductive definition enables to take into account global types containing 
branches, where some communications can be indefinitely procrastinated,  see Example~\ref{ue}. 
In order to do that, it 
is handy to associate to  a global type the set
of communication  labels which might (not necessarily) decorate 
  its 
 transitions. We dub them capabilities of the  global  type.

\begin{definition}[Capabilities]\label{cp}
Capabilities of global types are defined by: 

\Cline{\cp{\End}=\emptyset\quad\cp{\Sigma_{i\in I}\cml_i.\G_i}=\set{\cml_i}_{i\in I}\cup\bigcup_{i\in I}\cp{\G_i}}   
\end{definition}

\begin{definition}[LTS for global types]\label{ltsgt}
  The {\em labelled transition system (LTS) for global types} is
  specified by the following axiom and  
  rule:
  
  \Cline{
    \begin{array}{c}
      \NamedRule{\rn{E-comm}}{}{ \Sigma_{i\in I}\cml_i.\G_i \SLTS{\cml_j} \G_j}\quad \raisebox{4.5mm}{$j\in I$}\\[2mm]
     \NamedCoRule{\rn{I-comm}}{  \G_i \SLTS{\cml}\G'_i\quad\cml\in\cp{\G_i}
           \quad
      \plays{\cml}\cap\plays{\cml_i}=\emptyset
      \quad
            \forall i\in I 
      }{
     \Sigma_{i\in I}\cml_i.\G_i 
     \SLTS{\cml}
      \Sigma_{i\in I}\cml_i.\G'_i 
      }{}
      {}
    \end{array}}
\end{definition}

 Axiom \rn{E-comm} formalises the fact that, in a session exposing the behaviour
$\Sigma_{i\in I}\cml_i.\G_i $,
the communication labelled $\cml_j$ for any $j\in I$ can happen. 
When such a communication is actually performed, the resulting 
session will expose the behaviour $\G_j$.

 Rule \rn{I-comm}  enables  
to describe independent and concurrent communications, even if global types
apparently look like sequential descriptions of sessions' overall behaviours.
In fact, behaviours involving
participants ready to interact with each other
uniformly in all branches of a global type, can do  that 
if neither of them is involved
in an interaction appearing at top level in the global type.  The condition $\cml\in\cp{\G_i}$ in Rule \rn{I-comm} is needed  because such a rule is coinductive.
 In fact, without such a condition,  
 we could get 
 the following infinite derivation for $\G\SLTS{\pp\lambda\pq}\G$ with 
$\G=\pr\lambda'\ps.\G$:

\Cline{
\mathcal{D}= \prooftree
                       \mathcal{D}
                       \Justifies
                       \G\SLTS{\pp\lambda\pq}\G
                       \using \rn{I-comm}
                       \endprooftree
            } 
 
 \begin{example}[Use of coinduction in Rule \rn{I-comm}]\label{ue}      {\em     
            As shown in~\cite{BDL24} the coinductive formulation of Rule \rn{I-comm} allows to get $\G\SLTS{\pp\lambda\pq}\G'$, where 
$\G=\pr\lambda_1\ps.\G+\pr\lambda_2\ps.\pp\lambda\pq$  
and 
$\G'=\pr\lambda_1\ps.\G'+\pr\lambda_2\ps$. 
The inductive definition of this rule does not allow the shown transition. \finex}
\end{example}


\begin{example}[Typing the modular election session]
\begin{figure}[h!]
\prooftree
    \prooftree
              \prooftree
                    \prooftree
     \prooftree
         \prooftree
                  \prooftree
                            \prooftree
                                             \prooftree
                                                      \prooftree
                                                      \End \vdash \pP\pb\inact \parN \pP\pS\inact
                                                      \Justifies
                                              \widetilde{\G}'''_3  = \pS\msg{del}\pb 
                                  \vdash    
                                  \raisebox{-6mm}{
                    \scalebox{0.7}{ 
\begin{tikzpicture}[node distance=1cm]
        \node (pa) [draw, fill=lightgray,rectangle] {$\pa$};
        \node (pb) [draw,rectangle, below right of = pa,xshift = 7pt,yshift = 1pt] {$\pb$};
        \node (pc) [draw,rectangle, fill=lightgray, below left of = pb,xshift = 10pt,yshift = -10pt] {$\pc$};
        \node (pe) [draw, fill=lightgray, rectangle, below left of = pa,xshift = -7pt,yshift = 1pt] {$\pe$};       
        \node (pd) [draw,rectangle, fill=lightgray, below right of = pe,xshift = -10pt,yshift = -10pt] {$\pd$};
         \node (pS) [draw,rectangle, fill=green, below right of = pb,xshift = 10pt,yshift = 10pt] {$\pS$};
        \path  (pS) edge [-stealth] node {} (pb)
                  ;
 \end{tikzpicture}
 } }            
                                  \raisebox{-1mm}{\,\text{\Huge $\cdot\!\!\cdot\!\!\cdot$}}                \endprooftree
                                              \Justifies
                                               \widetilde{\G}''_3 = \pS\msg{gl}\pe.\widetilde{\G}'''_3 
                                               \vdash 
                                               \raisebox{-6mm}{
                    \scalebox{0.7}{ 
\begin{tikzpicture}[node distance=1cm]
        \node (pa) [draw, fill=lightgray,rectangle] {$\pa$};
        \node (pb) [draw,rectangle, below right of = pa,xshift = 7pt,yshift = 1pt] {$\pb$};
        \node (pc) [draw,rectangle, fill=lightgray, below left of = pb,xshift = 10pt,yshift = -10pt] {$\pc$};
        \node (pe) [draw, fill=green, rectangle, below left of = pa,xshift = -7pt,yshift = 1pt] {$\pe$};       
        \node (pd) [draw,rectangle, fill=lightgray, below right of = pe,xshift = -10pt,yshift = -10pt] {$\pd$};
         \node (pS) [draw,rectangle, fill=pink, below right of = pb,xshift = 10pt,yshift = 10pt] {$\pS$};
        \path  (pS) edge [-stealth,bend left=15] node {} (pe)
                  ;
 \end{tikzpicture}
 } }
                                              \raisebox{-1mm}{\,\text{\Huge $\cdot\!\!\cdot\!\!\cdot$}}
                                              \endprooftree
                           \Justifies
                  \widetilde{\G}'_3  =  \pgS\msg{del}\pw_2.\widetilde{\G}''_3 
                  \vdash
                           \raisebox{-5mm}{\scalebox{0.7}{ 
\begin{tikzpicture}[node distance=1cm]  
        \node (pw1) [draw,rectangle, fill=lightgray] {$\pw_1$};
        \node (pw2) [draw,rectangle, below right of = pw1,xshift = 0pt,yshift = -7pt] {$\pw_2$};
        \node (pw3) [draw,rectangle, fill=lightgray, below left of = pw1,xshift = 0pt,yshift = -7pt] {$\pw_3$};
        \node (gS) [draw,rectangle, fill=lime, above left of = pw3,xshift = 0pt,yshift =0pt] {$\pgS$};
        \path (gS) edge [-stealth]  node {} (pw2)
                  ;
 \end{tikzpicture}
 } }     
                  \raisebox{-1mm}{\,\text{\Huge $\cdot\!\!\cdot\!\!\cdot$}}
                           \endprooftree
                  \Justifies
                            \prooftree
                                     \prooftree
                          \widetilde{\G}_3 =  \pgS\msg{no}\pS_2.\widetilde{\G}'_3 
                          \vdash
                           \raisebox{-5mm}{\scalebox{0.7}{ 
\begin{tikzpicture}[node distance=1cm]  
        \node (pw1) [draw,rectangle, fill=lightgray] {$\pw_1$};
        \node (pw2) [draw,rectangle, below right of = pw1,xshift = 0pt,yshift = -7pt] {$\pw_2$};
        \node (pw3) [draw,rectangle, fill=lightgray, below left of = pw1,xshift = 0pt,yshift = -7pt] {$\pw_3$};
        \node (gS) [draw,rectangle, fill=green, above left of = pw3,xshift = 0pt,yshift =0pt] {$\pgS$};
        \node (pS1) [draw,rectangle, fill=Yellow, left of = gS,xshift = 0pt,yshift =0pt] {$\ptp{\pS}_{2}$};
        \path (gS) edge [-stealth]  node {} (pS1)
                  ;
 \end{tikzpicture}
 } }
                          \raisebox{-1mm}{\,\text{\Huge $\cdot\!\!\cdot\!\!\cdot$}}
                                     \Justifies
                           \pgS\msg{no}\pS_1.\widetilde{\G}_3 
                          \vdash
                          \raisebox{-5mm}{\scalebox{0.7}{ 
\begin{tikzpicture}[node distance=1cm]  
        \node (pw1) [draw,rectangle, fill=lightgray] {$\pw_1$};
        \node (pw2) [draw,rectangle, below right of = pw1,xshift = 0pt,yshift = -7pt] {$\pw_2$};
        \node (pw3) [draw,rectangle, fill=lightgray, below left of = pw1,xshift = 0pt,yshift = -7pt] {$\pw_3$};
        \node (gS) [draw,rectangle, fill=pink, above left of = pw3,xshift = 0pt,yshift =0pt] {$\pgS$};
        \node (pS1) [draw,rectangle, fill=Yellow, left of = gS,xshift = 0pt,yshift =0pt] {$\ptp{\pS}_1$};
        \path (gS) edge [-stealth]  node {} (pS1)
                  ;
 \end{tikzpicture}
 } }
                          \raisebox{-1mm}{\,\text{\Huge $\cdot\!\!\cdot\!\!\cdot$}}
                                      \endprooftree  
                             \Justifies
                                         \prooftree
                                              \hspace{32mm}\text{\Huge $\vdots$} \raisebox{1mm}{\footnotesize (local elections $1$ and $2$)}
                                               \Justifies
                  \G'_3 = \pgS\msg{gl}\pS.\G''_3
                  \vdash 
                  \raisebox{-5mm}{\scalebox{0.7}{ 
\begin{tikzpicture}[node distance=1cm]  
        \node (pw1) [draw,rectangle, fill=lightgray] {$\pw_1$};
        \node (pw2) [draw,rectangle, below right of = pw1,xshift = 0pt,yshift = -7pt] {$\pw_2$};
        \node (pw3) [draw,rectangle, fill=lightgray, below left of = pw1,xshift = 0pt,yshift = -7pt] {$\pw_3$};
        \node (gS) [draw,rectangle, fill=Yellow, above left of = pw3,xshift = 0pt,yshift =0pt] {$\pgS$};
        \node (pS3) [draw,rectangle, fill=Yellow, left of = gS,xshift = 0pt,yshift =0pt] {$\ptp{\pS}$};
        \path (gS) edge [-stealth]  node {} (pS3)
                  ;
 \end{tikzpicture}
 } }
                  \raisebox{-1mm}{\,\text{\Huge $\cdot\!\!\cdot\!\!\cdot$}}
                                                \endprooftree
                              \endprooftree
                  \endprooftree
         \hspace{2mm}
         \raisebox{-6mm}{\qquad
         } \quad
         \Justifies
\G_3 = \pw_3\msg{gl}\pgS.\G'_3 
\vdash 
\raisebox{-5mm}{\scalebox{0.7}{
\begin{tikzpicture}[node distance=1cm]  
        \node (pw1) [draw,rectangle, fill=lightgray] {$\pw_1$};
        \node (pw2) [draw,rectangle, below right of = pw1,xshift = 0pt,yshift = -7pt] {$\pw_2$};
        \node (pw3) [draw,rectangle, fill=Yellow, below left of = pw1,xshift = 0pt,yshift = -7pt] {$\pw_3$};
        \node (gS) [draw,rectangle, above left of = pw3,xshift = 0pt,yshift =0pt] {$\pgS$};
        \path (pw3) edge [-stealth,bend left=35]  node {} (gS)
                  ;
 \end{tikzpicture}
 } }
\raisebox{-1mm}{\,\text{\Huge $\cdot\!\!\cdot\!\!\cdot$}}
         \endprooftree  
\hspace{2mm}
\raisebox{-6mm}{$\mathcal{D}_2$}
\hspace{6mm}
\raisebox{-6mm}{$\mathcal{D}_1$}
\Justifies
\G'= \Sigma_{i\in\Set{1,2,3}}\pw_{i+1}\msg{l}\pw_i.\G_i
\vdash 
\raisebox{-6mm}{\scalebox{0.7}{
\begin{tikzpicture}[node distance=1cm]
        \node (pw1) [draw,rectangle] {$\pw_1$};
        \node (pw2) [draw,rectangle, below right of = pw1,xshift = 0pt,yshift = -7pt] {$\pw_2$};
        \node (pw3) [draw,rectangle, below left of = pw1,xshift = 0pt,yshift = -7pt] {$\pw_3$};
        \node (gS) [draw,rectangle, above left of = pw3,xshift = 0pt,yshift =0pt] {$\pgS$};
        \path  (pw3) edge [-stealth,bend right=35] node {} (pw2)
                  (pw2) edge [-stealth,bend right=35] node {} (pw1)
                  (pw1) edge [-stealth,bend right=35]  node {} (pw3)
                  ;
 \end{tikzpicture}
 } }
\raisebox{-1mm}{\,\text{\Huge $\cdot\!\!\cdot\!\!\cdot$}} 
\endprooftree
                    \Justifies
               \G''_{\pc\pe} = \pe\msg{e}\ptp{s}.\G'  \vdash
                    \raisebox{-6mm}{
                    \scalebox{0.7}{ 
\begin{tikzpicture}[node distance=1cm]
        \node (pa) [draw, fill=lightgray,rectangle] {$\pa$};
        \node (pb) [draw,rectangle, below right of = pa,xshift = 7pt,yshift = 1pt] {$\pb$};
        \node (pc) [draw,rectangle, fill=lightgray, below left of = pb,xshift = 10pt,yshift = -10pt] {$\pc$};
        \node (pe) [draw, fill=pink, rectangle, below left of = pa,xshift = -7pt,yshift = 1pt] {$\pe$};       
        \node (pd) [draw,rectangle, fill=lightgray, below right of = pe,xshift = -10pt,yshift = -10pt] {$\pd$};
         \node (pS) [draw,rectangle,  below right of = pb,xshift = 10pt,yshift = 10pt] {$\pS$};
        \path  (pe) edge [-stealth,bend right=15] node {} (pS)
                  ;
 \end{tikzpicture}
 }
 } \ \raisebox{-1mm}{\,\text{\Huge $\cdot\!\!\cdot\!\!\cdot$}}
                    \endprooftree
              \Justifies
               \G'_{\pd\pc} = \pc\msg{l}\pe.\G''_{\pc\pe} \vdash
              \raisebox{-6mm}{
              \scalebox{0.7}{ 
\begin{tikzpicture}[node distance=1cm]
        \node (pa) [draw, fill=lightgray,rectangle] {$\pa$};
        \node (pb) [draw,rectangle, below right of = pa,xshift = 7pt,yshift = 1pt] {$\pb$};
        \node (pc) [draw,rectangle, fill=Yellow, below left of = pb,xshift = 10pt,yshift = -10pt] {$\pc$};
        \node (pe) [draw, fill=Yellow, rectangle, below left of = pa,xshift = -7pt,yshift = 1pt] {$\pe$};       
        \node (pd) [draw,rectangle, fill=lightgray, below right of = pe,xshift = -10pt,yshift = -10pt] {$\pd$};
         \node (pS) [draw,rectangle,  below right of = pb,xshift = 10pt,yshift = 10pt] {$\pS$};
        \path  (pc) edge [-stealth] node {} (pe)
                  ;
 \end{tikzpicture}
 }
 }  \ \raisebox{-1mm}{\,\text{\Huge $\cdot\!\!\cdot\!\!\cdot$}}
              \endprooftree
              \quad
              \raisebox{-6mm}{$\mathcal{D}'_{\pc\pb}$} \quad
                  \Justifies
   \G_{\pa\pe} = \pd\msg{l}\pc.\G'_{\pd\pc} + \pc\msg{l}\pb.\G'_{\pc\pb} \vdash
   \raisebox{-6mm}{
\scalebox{0.7}{ 
\begin{tikzpicture}[node distance=1cm]
        \node (pa) [draw, fill=lightgray,rectangle] {$\pa$};
        \node (pb) [draw,rectangle, below right of = pa,xshift = 7pt,yshift = 1pt] {$\pb$};
        \node (pc) [draw,rectangle, below left of = pb,xshift = 10pt,yshift = -10pt] {$\pc$};
        \node (pe) [draw,fill=Yellow,rectangle, below left of = pa,xshift = -7pt,yshift = 1pt] {$\pe$};       
        \node (pd) [draw,rectangle,  below right of = pe,xshift = -10pt,yshift = -10pt] {$\pd$};
         \node (pS) [draw,rectangle,  below right of = pb,xshift = 10pt,yshift = 10pt] {$\pS$};
        \path  (pd) edge [-stealth,bend right=20] node {} (pc)
                  (pc) edge [-stealth,bend right=20] node {} (pb)
                  ;
 \end{tikzpicture}
 }
 }  \ \raisebox{-1mm}{\,\text{\Huge $\cdot\!\!\cdot\!\!\cdot$}}
    \endprooftree
    \hspace{0mm}
    \raisebox{-5mm}{$
    \mathcal{D}_{\pe\pd}
    \hspace{2mm}
    \mathcal{D}_{\pd\pc}
    \hspace{2mm}
    \mathcal{D}_{\pc\pb}
    \hspace{2mm}
    \mathcal{D}_{\pb\pa}$
    }
\Justifies
\pa\msg{l}\pe.\G_{\pa\pe} + \pe\msg{l}\pd.\G_{\pe\pd} + \pd\msg{l}\pc.\G_{\pd\pc} + \pc\msg{l}\pb.\G_{\pc\pb} + \pb\msg{l}\pa.\G_{\pb\pa} \vdash
\raisebox{-6mm}{
\scalebox{0.7}{ 
\begin{tikzpicture}[node distance=1cm]
        \node (pa) [draw,rectangle] {$\pa$};
        \node (pb) [draw,rectangle, below right of = pa,xshift = 7pt,yshift = 1pt] {$\pb$};
        \node (pc) [draw,rectangle, below left of = pb,xshift = 10pt,yshift = -10pt] {$\pc$};
        \node (pe) [draw,rectangle, below left of = pa,xshift = -7pt,yshift = 1pt] {$\pe$};       
        \node (pd) [draw,rectangle,  below right of = pe,xshift = -10pt,yshift = -10pt] {$\pd$};
        \node (pS) [draw,rectangle,  below right of = pb,xshift = 10pt,yshift = 10pt] {$\pS$};
        \path  (pa) edge [-stealth,bend right=20] node {} (pe)
                  (pe) edge [-stealth,bend right=20] node {} (pd)
                  (pd) edge [-stealth,bend right=20] node {} (pc)
                  (pc) edge [-stealth,bend right=20] node {} (pb)
                  (pb) edge [-stealth,bend right=20] node {} (pa)
                  ;
 \end{tikzpicture}
 }
 } \ \raisebox{-1mm}{\,\text{\Huge $\cdot\!\!\cdot\!\!\cdot$}}
\endprooftree
\vspace{4mm}
\caption{A type derivation for Example \ref{ex:gelection}.}\label{fig:typingex} 
\end{figure}
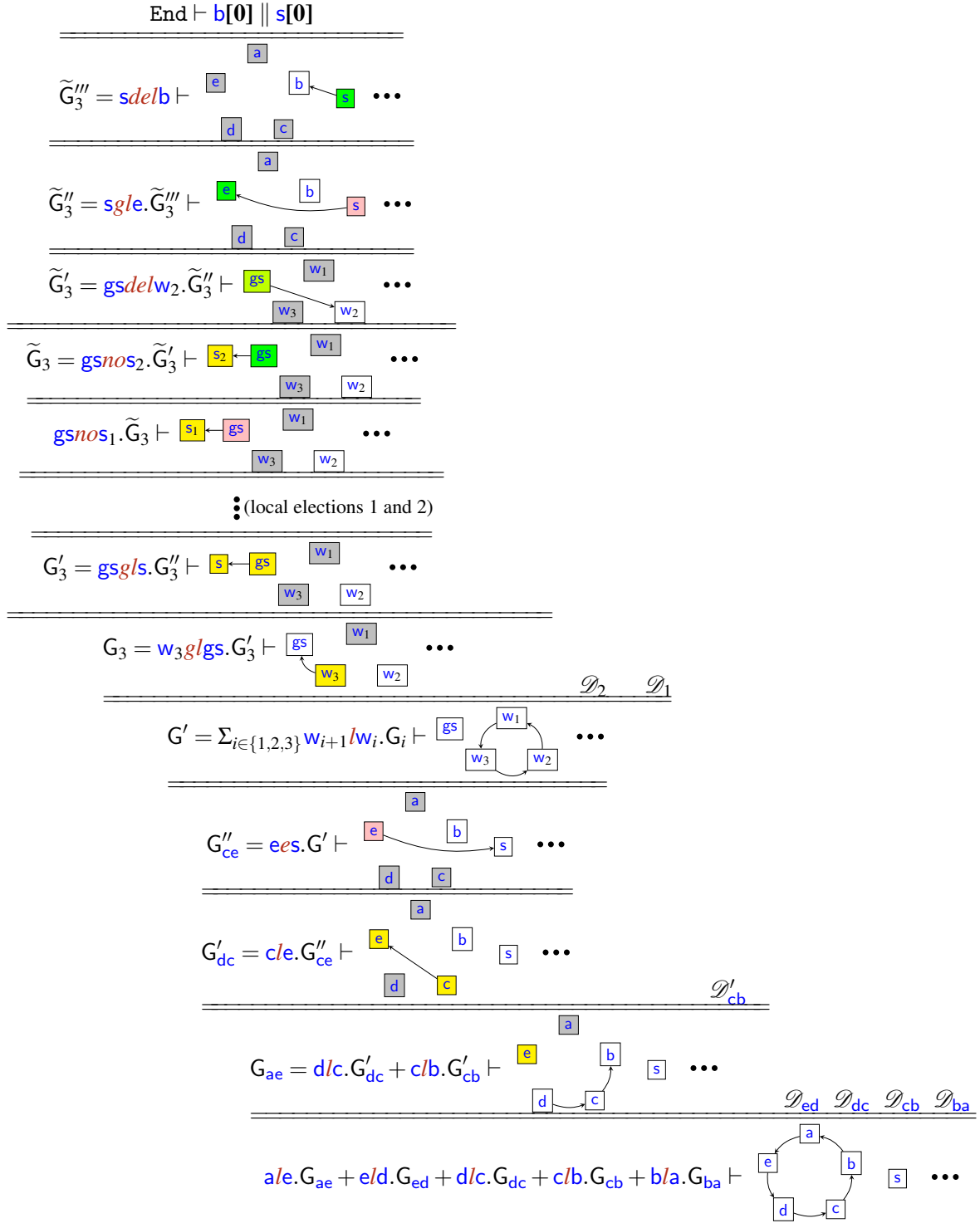
{\em
In Figure \ref{fig:typingex} we provide a typing for the modular session $\eg$ of Example \ref{ex:gelection} in the type system $\vdash^{\mathcal{P}}$ with 

  \Cline{{\mathcal{P}}=\bigcup_{i\in  \set{1,2,3}  
  }\set{\set{\pa_i,\pb_i,\pc_i,\pd_i,\pe_i,\pS_i}} \cup\Set{\set{\pw_1,\pw_2,\pw_3,\pgS}}}
  
  \smallskip
  
  \noindent 
 We  use colours for representing reduced participants.
 In particular, a participant has the colour 
 
 \noindent $\begin{tikzpicture}
 \node (pw3) [draw,rectangle] {};
 \end{tikzpicture}$  in its initial state; \qquad \qquad\qquad
 $\begin{tikzpicture}
 \node (pw3) [draw,rectangle,fill=Yellow] {};
\end{tikzpicture}$  after one interaction; \qquad\qquad\qquad
$\begin{tikzpicture}
 \node (pw3) [draw,rectangle,fill=pink] {};
\end{tikzpicture}$ after two interactions;\\
$\begin{tikzpicture}
 \node (pw3) [draw,rectangle,fill=green] {};
\end{tikzpicture}$ after three interactions;\qquad\qquad
$\begin{tikzpicture}
 \node (pw3) [draw,rectangle,fill=lime] {};
\end{tikzpicture}$ after four interactions;\qquad\qquad\quad~
$\begin{tikzpicture}
 \node (pw3) [draw,rectangle,fill=lightgray] {};
\end{tikzpicture}$  when terminated. 

\smallskip

\noindent
For instance:\\
\raisebox{-1mm}{\scalebox{0.7}{$\begin{tikzpicture}
 \node (pa) [draw,rectangle] {$\pgS$};
\end{tikzpicture}$}}  $= 
\pP{\pgS}{\Sigma_{i\in\Set{1,2,3}}\pw_i?\msg{gleader}.\pS_i!\msg{gleader}.\pS_{i+1}!\msg{no}.\pS_{i+2}!\msg{no}.\big(  \Sigma_{i\in\Set{1,2,3}}\pw_i!\msg{del}\big)}$\\
\raisebox{-1mm}{\scalebox{0.7}{$\begin{tikzpicture}
 \node (pa) [draw,rectangle,fill=Yellow] {$\pgS$};
\end{tikzpicture}$}}  $= 
\pP{\pgS}{\pS_3!\msg{gleader}.\pS_{1}!\msg{no}.\pS_{2}!\msg{no}.\big(  \Sigma_{i\in\Set{1,2,3}}\pw_i!\msg{del}\big)}$\qquad\qquad
\raisebox{-1mm}{\scalebox{0.7}{$\begin{tikzpicture}
 \node (pa) [draw,rectangle,fill=pink] {$\pgS$};
\end{tikzpicture}$}}  $= 
\pP{\pgS}{\pS_{1}!\msg{no}.\pS_{2}!\msg{no}.\big(  \Sigma_{i\in\Set{1,2,3}}\pw_i!\msg{del}\big)}$\\
\raisebox{-1mm}{\scalebox{0.7}{$\begin{tikzpicture}
 \node (pa) [draw,rectangle,fill=green] {$\pgS$};
\end{tikzpicture}$}}  $= 
\pP{\pgS}{\pS_{2}!\msg{no}.\big(  \Sigma_{i\in\Set{1,2,3}}\pw_i!\msg{del}\big)}$\qquad\qquad
\raisebox{-1mm}{\scalebox{0.7}{$\begin{tikzpicture}
 \node (pa) [draw,rectangle,fill=lime] {$\pgS$};
\end{tikzpicture}$}}  $= 
\pP{\pgS}{ \Sigma_{i\in\Set{1,2,3}}\pw_i!\msg{del}}$\qquad\qquad
\raisebox{-1mm}{\scalebox{0.7}{$\begin{tikzpicture}
 \node (pa) [draw,rectangle,fill=lightgray] {$\pgS$};
\end{tikzpicture}$}}  $= 
\pP{\pgS}{\inact}$

\bigskip
 
 In Figure \ref{fig:typingex},
arrows do connect pair of participants forming a redex. 
 Moreover, 
in the conclusion of the rules we show only the participants of the
module containing the coherent set of reductions and, in case, the ``external'' connectors.
The rest of the session will be  denoted by ``\raisebox{-1mm}{\,\text{\Huge $\cdot\!\!\cdot\!\!\cdot$}}''. 
 Also, the figure 
shows only one  
branch of the  typing  derivation tree:  the one concerning the 
global election of $\pe_3$.

For the sake of readability, $\msg{l}$ and $\msg{gl}$ are abbreviation for, respectively,
 $\msg{leader}$ and $\msg{gleader}$. Moreover, $\pa,\pb,\pc,\pd,\pe$ and $\pS$ stand
 for $\pa_3,\pb_3,\pc_3,\pd_3,\pe_3$ and $\pS_3$.
\finex
} 

\end{example}

\section{Properties}\label{spr}

A  subsession  
of the shape  $\pP{\pp}{\q!\msg{\lambda}.\PP\pplus\PP'} \parN \pP{\q}{\pp?\msg{\lambda}.\Q\pplus\Q'}$  is called a {\em redex} and 
$\pP{\pp}{\PP} \parN \pP{\q}{\Q}$  
is the {\em contractum} of the redex. In a transition labelled by ${\pp}{\msg{\lambda}}{\q}$ both the redex
and the contractum are uniquely determined.

\begin{lemma}\label{lem:unique-redex}
If $\Nt \SLTS{{\pp}{\msg{\lambda}}{\q}} \Nt'$, then there exists a unique redex
$\pP{\pp}{\q!\msg{\lambda}.\PP \pplus \PP'} \parN \pP{\q}{\pp?\msg{\lambda}.\Q \pplus \Q'}$
such that  

\Cline{
\Nt \equiv \pP{\pp}{\q!\msg{\lambda}.\PP \ppplus \PP'} \parN \pP{\q}{\pp?\msg{\lambda}.\Q \ppplus \Q'} \parN \Nt'' 
}

\noindent
and 
$\Nt' \equiv \pP{\pp}{\PP} \parN \pP{\q}{\Q} \parN \Nt''$.
\end{lemma}
\begin{proof}
Immediate by the definition of session LTS. 
\qed
\end{proof}

\medskip

Rule \rn{Comm} in Definition~\ref{slts} entails 
an easy relation between the participants  connected by reductions
in a session.

\begin{lemma}\label{pr}
If $\Nt\SLTS{\cml}\Nt'$, then $\plays\Nt=\plays\cml\cup\plays{\Nt'}$.
\end{lemma}

It is not difficult to check that the participants of a session and of its global type are the same. 
\begin{lemma}\label{l:p}
If \, $\tyn\G\Nt$, then $\plays\G=\plays\Nt$.
\end{lemma}

  The following technical lemma relating capabilities and possible reductions 
 of a global type will be handy later on. 
\begin{lemma}\label{lem:capabilities}
If \hspace{1pt}$\G \SLTS{\cml} \G'$, then $\cml \in \cp{\G}$.
\end{lemma}
\begin{proof}
By cases  on the applied  axiom/rule justifying 
$\G \SLTS{\cml} \G'$.  If  this is $\rn{E-Comm}$, then $\G = \Sigma_{i\in I}\cml_i.\G_i$ and $\cml=\cml_j$
for some $j \in I$ and $\cml_j \in \cp{\Sigma_{i\in I}\cml_i.\G_i}$ by Definition~\ref{cp}. 

\noindent
Otherwise,  $\G = \Sigma_{i\in I}\cml_i.\G_i$ and $\G'= \Sigma_{i\in I}\cml_i.\G_i'$  by Rule $\rn{I-Comm}$, where
$\G_i \SLTS{\cml} \G'_i$, $\cml\in\cp{\G_i}$ 
 and $\plays{\cml}\cap\plays{\cml_i}=\emptyset$ for all $i \in I$. 
This implies 
$\cml \in \cp{\G}$,  since $\bigcup_{i \in I} \cp{\G_i}\subseteq\cp\G$ by Definition~\ref{cp}. 
\qed
\end{proof}

\smallskip
 For showing Subject Reduction it is crucial to ensure that the $\mathcal P$-coherence of a set of labels is preserved by reducing a label not belonging to this set, see 
 Lemma~\ref{lem:coherpres} whose proof uses Lemma~\ref{lem:cohint}  below. 

\begin{lemma}
\label{lem:cohint}
Let $\set{\cml_i}_{i\in I}$ be $\mathcal P$-coherent for $\Nt$ and let $\cml\in\LL\Nt$.
Moreover, let $\cml\neq \cml_i$ for all $i\in I$.  Then $\plays{\cml} \,\cap\,\plays{\cml_i}= \emptyset\,$ for all $i\in I$. 
\end{lemma}
\begin{proof}
 By definition of coherence  (Definition~\ref{csrl}),  we have a subsession  $\widehat\Nt$  of $\Nt$
witnessing 
the $\mathcal P$-coherence of $\set{\cml_i}_{i\in I}$.
By contradiction, let us assume $\plays{\cml} \,\cap\,\plays{\cml_j}\neq \emptyset$ for some 
$j\in I$. By definition of $\mathcal P$-modularisation, this implies that $\plays{\cml_j}$ contains a connector $\pp$ of $\widehat\Nt$.
Let $\plays{\cml_j}=\set{\pp,\pq}$ with $\pq\not\in\plays{\widehat\Nt}$ and $\plays{\cml}=\set{\pq,\pr}$. 
Then $\q\in\plays{\Nt}$ and the process $\PQ$ of $\pq$
must have a choice between a communication with $\pp$
 and a communication 
with $\pr$. But this is impossible, since $\pq$ must be a connector for some subsession $\widehat\Nt'$ of $\Nt$  by condition~(\ref{modularisation2}) of Definition~\ref{modularisation} and then the process $\PQ$ must be $\plays{\widehat\Nt'}$-connecting by Definition~\ref{def:conn}. 
\qed
\end{proof}

\begin{lemma}[Coherence preservation]
\label{lem:coherpres}
Let $\set{\cml_i}_{i\in I}$ be $\mathcal P$-coherent for $\Nt$ and let $\Nt \SLTS\cml \Nt'$.
Then $\cml\not\in\set{\cml_i}_{i\in I}$ implies that $\set{\cml_i}_{i\in I}$ is $\mathcal P$-coherent for $\Nt'$
as well.
\end{lemma}
\begin{proof} Let $\widehat\Nt$ be the subsession of $\Nt$ witnessing the $\mathcal P$-coherence of $\set{\cml_i}_{i\in I}$ for $\Nt$. From $\cml\not\in\set{\cml_i}_{i\in I}$ and 
 Lemma \ref{lem:cohint} we get that 
$\plays{\cml} \,\cap\,\plays{\cml_i}= \emptyset$ for all $i\in I$.
This implies that the reduction $\SLTS{\cml}$ cannot affect any reduction with label in 
$\set{\cml_i}_{i\in I}$. Hence 
$\widehat\Nt$ is a witness of the $\mathcal P$-coherence of $\set{\cml_i}_{i\in I}$ also for $\Nt'$.
\qed
\end{proof} 
\smallskip

Notice how the conditions on connectors (Definitions \ref{cpr} and \ref{def:conn}) are crucial in getting
the property that $\plays{\cml} \,\cap\,\plays{\cml_i}= \emptyset$ for all $i\in I$ in the above 
result of coherence preservation. If we allowed connectors to communicate with external partners 
having unrestricted processes we could
consider the $\Set{\Set{\pp},\Set{\pr,\ps}}$-modularisation of
$\Nt= \pP\pp{\pr!\msg{\lambda}} \parN \pP\pr{\pp?\msg{\lambda}+\ps?\msg{\lambda}}
\parN \pP\ps{\pr!\msg{\lambda}}$.
In such a case, the set $\Set{\pp\msg{\lambda}\pr}$ would be $\Set{\Set{\pp},\Set{\pr,\ps}}$-coherent with witness $\Nt'= \pP\pp{\pr!\msg{\lambda}}$. 
However, we would also have that $\Nt\SLTS{\ps\msg{\lambda}\pr} \Nt'$,  
but $\Set{\pp\msg{\lambda}\pr}$ would not be $\Set{\Set{\pp},\Set{\pr,\ps}}$-coherent for $\Nt'$,  since $\pp\msg{\lambda}\pr\not\in\LL{\Nt'}$.

\begin{theorem}[Subject Reduction]\label{thm:SR}
 If  \hspace{1pt}$\tyn\G  \Nt $ and $\Nt \SLTS\cml \Nt'$, then $ \tyn{\G'} \Nt'$
 and  $\G \SLTS\cml \G' $ 
   for some $\G'$. 
\end{theorem}
\noindent
\begin{proof} By coinduction on the derivation of $\tyn\G \Nt$. Let $\cml=\pp\msg{\lambda}\q$.
By Lemma \ref{lem:unique-redex},
if $\Nt \SLTS\cml \Nt'$, then there exists a unique redex

\Cline{
\RR = \pP{\pp}{\q!\msg{\lambda}.\PP \pplus \PP'}\parN\pP{\q}{\pp?\msg{\lambda}.\Q \pplus \Q'}
}

\noindent
such that $\Nt \equiv \RR \parN \Nt''$ and 
$\Nt' \equiv \pP{\pp}{\PP} \parN \pP{\q}{\Q} \parN \Nt''$ for some $\Nt''$.
By the hypothesis that $\tyn\G \Nt$ we know that 
$\G$ is of the form $\Sigma_{i\in I}\cml_i.\G_i$ and the derivation ends by

	\Cline{
	\NamedCoRule{\rn{TComm}}
{\mbox{$\begin{array}{c} \Nt\SLTS{\cml_i}\Nt_i\qquad\tyn{\G_i}{\Nt_i}\qquad\forall i\in I\neq\emptyset\\
\set{\cml_i}_{i\in I}\text{ is $\mathcal P$-coherent for } \Nt \qquad\plays{\Sigma_{i\in I}\cml_i.\G_i}=\plays\Nt\\[3mm]
\end{array}$}}
{\tyn{\Sigma_{i\in I}\cml_i.\G_i}{\Nt}}
{ }{}
	}

\noindent
We proceed by distinguishing the two possible following cases.

{\em Case}  $\pp\msg{\lambda}\q=\cml_j$ for some $j\in I$. 
 By the premises of the rule, we have $\Nt \SLTS{\cml} \Nt_j$ and $\tyn{\G_j} \Nt_j$, where
$\Nt'=\Nt_j$ .  
Moreover, it immediately follows that $\G \SLTS{\cml} \G_j$  by  Axiom  $\rn{E-Comm}$.

{\em Case}   $\pp\msg{\lambda}\q\neq\cml_i$ for all $i\in I$.  
We have that, for all $i\in I$,  $\Nt_i\equiv\RR\parN\Nt'_i$ for some $\Nt'_i$. 
Hence we get that $\Nt_i \SLTS{\cml} \pP{\pp}{\PP} \parN \pP{\q}{\Q} \parN \Nt'_i$ for all $i\in I$. Moreover, for all $i\in I$,
$\Nt'' \SLTS{\cml_i}\Nt'_i$. 
	 By the  coinduction hypothesis  on the premises of the rule, we have that, for all $i \in I$, $\tyn{\G'_i}{\pP{\pp}{\PP} \parN \pP{\q}{\Q} \parN \Nt'_i}$ and 
	 $\G_i \SLTS{\cml} \G'_i$ for some $\G'_i$. 
	Now, by Lemma \ref{lem:capabilities}, we get that $\cml \in \cp{\G_i}$ for all $i \in I$, hence 	 we have that
	 $\G \SLTS{\cml}  \Sigma_{i\in I}\cml_i.\G_i'$  by Rule $\rn{I-Comm}$.  From $\Nt\SLTS\cml\Nt'\equiv \pP{\pp}{\PP} \parN \pP{\q}{\Q} \parN \Nt''$ and 
	 $\Nt'' \SLTS{\cml_i}\Nt'_i$  
	we get $\Nt'\SLTS{\cml_i}\pP{\pp}{\PP} \parN \pP{\q}{\Q} \parN \Nt'_i$ for all $i\in I$,    which imply,  by Lemma~\ref{pr},
	
	\Cline{\plays{\Nt'}=\bigcup_{i\in I}\plays{\cml_i}\cup\bigcup_{i\in I}\plays{\pP{\pp}{\PP} \parN \pP{\q}{\Q} \parN \Nt'_i}}  
	
	\noindent By Lemma~\ref{l:p},
	$\tyn{\G'_i}{\pP{\pp}{\PP} \parN \pP{\q}{\Q} \parN \Nt'_i}$
	  gives 
	$\plays{\G'_i}=\plays{\pP{\pp}{\PP} \parN \pP{\q}{\Q} \parN \Nt'_i}$  for all $i\in I$.  
	Hence 
	$\plays{\Sigma_{i\in I}\cml_i.\G_i'}= \bigcup_{i\in I}\plays{\cml_i}\cup\bigcup_{i\in I}\plays{\G'_i}=\plays{\Nt'}$. 
	Moreover, 
	from Lemma \ref{lem:coherpres} and $\pp\msg{\lambda}\q\neq\cml_i$ for all $i\in I$, it follows that  $\set{\cml_i}_{i\in I}$ is $\mathcal P$-coherent for $\Nt'$ as well.  Therefore Rule \rn{TComm} applies, namely
	
\Cline{\qquad\qquad\qquad\quad
	\NamedCoRule{}
{\mbox{$\begin{array}{c} \Nt'\SLTS{\cml_i}\pP{\pp}{\PP} \parN \pP{\q}{\Q} \parN\Nt'_i\quad\tyn{\G'_i}{\pP{\pp}{\PP} \parN \pP{\q}{\Q} \parN\Nt'_i}\quad\forall i\in I\neq\emptyset\\
\set{\cml_i}_{i\in I}\text{ is $\mathcal P$-coherent for }  \Nt' 
\qquad \plays{\Sigma_{i\in I}\cml_i.\G'_i}=\plays{\Nt'}\\[2mm]
\end{array}$}}
{\tyn{\Sigma_{i\in I}\cml_i.\G'_i}{\Nt'}}
{ }{}
	\quad\qquad\qquad\qquad\text{\qed}}
\end{proof}

\medskip

\begin{theorem}[Session Fidelity]\label{thm:SF}
 If \hspace{1pt}$\G \vdash \Nt $ and $\G \SLTS{\cml} \G' $, then 
 $\Nt \SLTS{\cml} \Nt'$ and  $ \G' \vdash \Nt'$ 
  for some $\Nt'$. 
\end{theorem}
\begin{proof}
By coinduction on the derivation of  $\G \SLTS{\cml} \G'$. We distinguish two cases according to the  axiom/rule  justifying $\G \SLTS{\cml} \G'$.

 {\em Axiom}  $\rn{E-Comm}$: then  $\G = \Sigma_{i\in I}\cml_i.\G_i$, 
 $\cml = \cml_j$ and $\G'=\G_j$ for some $j \in I$.  Since $\G \neq \End$, the last rule in the derivation of $\G \vdash \Nt$
	must be $\rn{TComm}$, which implies that 
	 $\cml=\Act{\pp}{\msg{\lambda}}{\q}$
	for some ${\pp}$, ${\msg{\lambda}}$ and ${\q}$ such that 
	
	\Cline{
	\Nt \equiv \pP{\pp}{\q!\msg{\lambda} .\PP \pplus  \PP'} \parN \pP{\q}{\pp?\msg{\lambda}.\PQ \pplus  \Q'} \parN \Nt_0 \SLTS{\Act{\pp}{\msg{\lambda}}{\q}} {\pP{\pp}{\PP}} \parN \pP{\q}{\Q} \parN \Nt_0\equiv\Nt' 
	}
	
	\noindent
	 for some $\Nt_0$, and $\tyn{\G'}{\Nt'}$.

{\em Rule}  $\rn{I-Comm}$: then $\G =  \Sigma_{i\in I}\cml_i.\G_i$  and $\G'=  \Sigma_{i\in I}\cml_i.\G'_i $ 
with $\G_i \SLTS{\cml}\G'_i$ and $\cml\in\cp{\G_i}$ and $\plays{\cml}\cap\plays{\cml_i}=\emptyset$ for all $i\in I$. 

 \noindent
	 Since the last rule in the derivation of $\G \vdash \Nt$
	must be $\rn{TComm}$,
	it follows that
	\begin{itemize}
	\item 
$\set{\cml_i}_{i\in I}\text{ is $\mathcal P$-coherent for } \Nt$;
      \item 	
$\Nt\SLTS{\cml_i}\Nt_i$ and $\tyn{\G_i}{\Nt_i}$, for all $ i\in I\neq\emptyset$;
      \item
$\plays{\Sigma_{i\in I}\cml_i.\G_i}=\plays\Nt$.
	\end{itemize}

\noindent
	By the coinduction hypothesis, we know that, for each $i\in I$, there exists $\Nt'_i$ such that
	
	\Cline{
	 \Nt_i \SLTS{\cml} \Nt'_i \qquad \mbox{and} \qquad \G'_i \vdash \Nt'_i
	}
	
\noindent
	Notice that, being the label $\cml$ the same for all these reductions, by Lemma \ref{lem:unique-redex} there exists a unique redex 
	
	\Cline{
	\pP{\pp}{\q!\msg{\lambda}.\PP\ \pplus\ \PP'} \parN \pP{\q}{\pp?\msg{\lambda}.{\Q}\ \pplus\ \Q'}
	} 
	
	\noindent
	with contractum $\pP{\pp}{\PP} \parN \pP{\q}{\Q}$  in all the $\Nt_i$, 
	 such  that
	$\cml =\Act{\pp}{\msg{\lambda}}{\q}$.
	On the other hand, since we know that 
	 $\plays{\cml}\cap\plays{\cml_i}=\emptyset$ for all $i\in I$, 
     it must be the case that   $\cml_i=\pr_i\msg{\lambda}_i\ps_i$ and  
    
     \Cline{\Nt_i' 
     \equiv \pP{\pr_i}{R_i} \parN \pP{\ps_i}{S_i} \parN \Nt_i''}
     
\noindent
	for  some $\pr_i$, $\ps_i$, $\PR_i$, $S_i$, $\Nt_i''$ and for all $i\in I$.
Hence, since  
	 $\Nt_i \SLTS{\Act{\pp}{\msg{\lambda}}{\q}} \Nt'_i$,  we have that, for each $i\in I$,
	
	\Cline{
	\begin{array}{c}
		\Nt \equiv {\pP{\pr_i}{\ps_i!\msg{\lambda}_i.R_i\ \pplus\ R'_i}} \parN \pP{\ps_i}{\pr?\msg{\lambda}_i.S_i\ \pplus\ S'_i} \parN \Nt_i'' 
		\SLTS{\Act{\pp}{\msg{\lambda}}{\q}} 
		{\pP{\pr_i}{\ps_i!\msg{\lambda}_i.R_i\ \pplus\ R'_i}} \parN \pP{\ps_i}{\pr?\msg{\lambda}_i.S_i\ \pplus\ S'_i} \parN \Nt_i''' \equiv \Nt'
	\end{array}
	}	
	
\noindent
 for some $\PR'_i$, $S'_i$ (if any), $\Nt_i'''$   and for all $i\in I$. 
	
By Lemma~\ref{l:p}, $\G_i'\vdash\pP{\pr_i}{R_i} \parN \pP{\ps_i}{S_i} \parN \Nt_i''$ implies  
$\plays{\G'_i}=\plays{\pP{\pr_i}{R_i} \parN \pP{\ps_i}{S_i} \parN  \Nt_i''}$  
 and then $\plays{\G'}=\plays{\Nt'}$. 
 Moreover, from $\plays{\cml}\cap\plays{\cml_i}=\emptyset$ for all $i\in I$, 
 we immediately get that $\cml\not\in\Set{\cml_i}_{i\in I}$. 
 So, by Lemma \ref{lem:coherpres} we get that $\set{\cml_i}_{i\in I}\text{ is $\mathcal P$-coherent for } \Nt'$.
	 We conclude that there exists  a  
	derivation ending by  the following application of  Rule $\rn{TComm}$
	
	\Cline{\qquad\qquad
\NamedCoRule{\rn{TComm}}
{\mbox{$\begin{array}{c}  \Nt' 
\SLTS{\cml_i} \Nt_i'\qquad\G_i'\vdash\Nt_i'\qquad\forall i\in I\neq\emptyset\\[0.5mm]
\set{\cml_i}_{i\in I}\text{ is $\mathcal P$-coherent for } \Nt' \qquad\plays{\G'}=\plays{\Nt'} \\[3mm]
\end{array}$}}
{\tyn{\G'}{\Nt'}}
{ }{}
\qquad\qquad\text{\qed}}	
\end{proof}

\medskip
Toward establishing the property that typable sessions are lock free, we first prove the following lemma. In words, 
if $\pp \in \plays{\G}$, then it must occur somewhere in its syntactic tree, hence there is a trace $\sigma\cdot \cml$ out of $\G$, consisting just
of external communications, which corresponds to  a 
path in the tree ending by  the first 
communication label $\cml$ involving $\pp$.

 \begin{lemma}\label{lem:G-participants} If $\pp \in \plays{\G}$, then there are  $\sigma$,  $\cml$ and $\G'$ such that $ \G  \SLTS{\;\sigma\cdot\cml\;}  \G'$, 
 $\pp \not\in \plays{\sigma}$   and $\pp \in \plays{\cml}$.
\end{lemma}
\begin{proof} The proof is by coinduction on $\G$. 
Since $\pp \in \plays{\G}$ we have that $\G = \Sigma_{i\in I}\cml_i.\G_i$.	
Now, let us assume $\pp \in \bigcup_{i\in I}\plays{\cml_i}$.
Without loss of generality, we can also assume that $\pp\in\plays{\cml_j}$ for some
$j\in I$.
Then  we immediately have that $\G \SLTS{\;\cml_j\;} \G_j$
by  Axiom $\rn{E-Comm}$,  and the thesis trivially
follows by taking $\sigma=\varepsilon$.  
Otherwise, since $\pp \in \plays{\G} = \bigcup_{i\in I}\plays{\cml_i} \cup \bigcup_{i\in I}\plays{\G_i}$, 
we have that $\pp \not\in \bigcup_{i\in I}\plays{\cml_i}$ implies $\pp \in \plays{\G_j}$ for some $j\in I$. 
By  the  coinduction hypothesis, we have that there are  a  $\sigma'$ and a  $\cml$   
such that  $\G_j  \SLTS{\;\sigma'\cdot  \cml  \;}  \G'$, 
$\pp\not\in\plays{\sigma'}$ and $\pp\in\plays{\cml}$.
Then the thesis follows by setting  
$\sigma=\cml_j\cdot\sigma'$, since $\G\SLTS{\cml_j}\G_j$ by  Axiom  \rn{E-Comm} and $\G_j\SLTS{\sigma'\cdot\cml}\G'$. 
\qed
\end{proof}

\medskip
Observe that the last lemma is a sort of inverse implication w.r.t. Lemma \ref{lem:capabilities}, 
 since it shows that
the existence of a capability which is an actual communication of a global type $\G$ follows by the fact that one of the involved participants is in $\plays{\G}$.

We are now in place to prove that typable sessions are lock free.

\begin{theorem}[Lock Freedom]
 If \hspace{1pt}$\Nt$ is typable, then $\Nt$ is lock free.
\end{theorem}
\begin{proof}
Let $\tyn\G\Nt$.  Following Definition \ref{d:lf}, in order to prove Lock Freedom for $\Nt$, 
let $\Nt \SLTS{\sigma}\Nt'$ for a finite $\sigma$ and let $\pp\in\plays{\Nt'}$.
By  Subject Reduction (Theorem \ref{thm:SR}) we get   $\tyn{\G'}{\Nt'}$.
We can now recur to Lemma~\ref{l:p} and get $\pp \in  \plays{\G'}$.
From the fact that $\pp \in \plays{\G'}$ and by Lemma \ref{lem:G-participants}
it follows that $\G'  \SLTS{\sigma'\cdot\cml\;}  \G''$ for some  $\sigma'$ and  $\cml$ with  $\pp \not\in \plays{\sigma'}$ and $\pp \in \plays{\cml}$. Now the thesis follows
by Session Fidelity (Theorem \ref{thm:SF}). 
\qed
\end{proof}

\medskip

We conclude this section by showing that typability of a session does depend on the choice 
neither of the $\mathcal P$-coherent sets nor of the partition $\mathcal P$.

\begin{theorem}\label{tpp}
If  $\Nt$ is typable in $\vdash^{\mathcal P}$ and $\set{\cml_i}_{i\in I}$ is $\mathcal P$-coherent for $\Nt$, then there are $\G_i$ for $i \in I$ such that
$\tyn{\Sigma_{i\in I}\cml_i.\G_i}\Nt$. 
\end{theorem}
\begin{proof}   The $\mathcal P$-coherence of $\set{\cml_i}_{i\in I}$ for $\Nt$ gives 
$\Nt\SLTS{\cml_i}\Nt_i$, which implies 
by Subject Reduction (Theorem~\ref{thm:SR}) $\tyn{\G_i}{\Nt_i}$ for some $\G_i$ and all $i\in I$. 
By Lemma~\ref{l:p} $\plays{\G_i}=\plays{\Nt_i}$ for all $i\in I$. 
From $\Nt\SLTS{\cml_i}\Nt_i$ for all $i\in I$ we get  $\plays\Nt=\bigcup_{i\in I}\plays{\cml_i}\cup \bigcup_{i\in I}\plays{\Nt_i}$ by Lemma~\ref{pr}. 
By definition, $\plays{\Sigma_{i\in I}\cml_i.\G_i}=\bigcup_{i\in I}\plays{\cml_i}\cup \bigcup_{i\in I}\plays{\G_i}$.  We conclude $\plays{\Sigma_{i\in I}\cml_i.\G_i}=\plays{\Nt}$, so we can derive $\tyn{\Sigma_{i\in I}\cml_i.\G_i}\Nt$  using 
 Rule \rn{TComm}.  
\qed
\end{proof}

\begin{theorem}\label{ttpp} If $\Nt$ is typable in $\vdash^{\mathcal P}$ and it is ${\mathcal P'}$-modularisable, then $\Nt$ is typable in $\vdash^{\mathcal P'}$ too.
\end{theorem}
\begin{proof} Since $\Nt$ is typable in $\vdash^{\mathcal P}$, then $\Nt$ is ${\mathcal P}$-modularisable by Definition~\ref{csrl}. 
Let $\mathcal P=\set{\pSet_k}_{k\in K}$ and $\set{\cml^{ k }_h}_{h\in{H_k}}=\set{\cml\in\LL\Nt\mid\plays\cml\cap\pSet_k\neq\emptyset}$.
By definition of partition we observe that $\bigcup_{k\in K}\bigcup_{h\in H_k}\cml^{ k }_h
=\LL\Nt$.  By Definitions~\ref{csrl} and~\ref{modularisation} for all $k\in K$ $\set{\cml^{ k }_h}_{h\in{H_k}}$ is ${\mathcal P}$-coherent for $\Nt$.
By Lemma~\ref{tpp} for each $k\in K$ there is a global type $\G_k=\Sigma_{h\in H_k}\cml^{ k }_h.\widehat\G_h$  which can be assigned to $\Nt$ in $\vdash^{\mathcal P}$,
that is $\tyn{\Sigma_{h\in H_k}\cml^{ k }_h.\widehat\G_h}{\Nt}$. 

\noindent
Let $\set{\cml_i}_{i\in I}$ be $\mathcal P'$-coherent for $\Nt$. 
We proceed now by coinduction simultaneously on the derivations 
$\tyn{\Sigma_{h\in H_k}\cml^{ k }_h.\widehat\G_h}{\Nt}$  for all $k\in K$. 
From the above, for each $i\in I$ there is $k_i\in K$ and $l_i\in H_{k_i}$
such that one of the premises of the conclusion
 $\tyn{\Sigma_{h\in H_{k_i}}\cml^{ k_i }_h.\widehat\G_h}{\Nt}$
is $\tyn{\widehat\G_{l_i}}{\Nt_i}$ where $\Nt\SLTS{\cml_i}\Nt_i$.  
By coinduction we get $\G''_i\vdash^{\mathcal P'}\Nt_i$ for some $\G''_i$ and for all $i\in I$. 
By Lemma~\ref{l:p} $\plays{\G''_i}=\plays{\Nt_i}$ for all $i\in I$.  We have $\plays{\Sigma_{i\in I}\cml_i.\G''_i}=\bigcup_{i\in I}\plays{\cml_i}\cup\bigcup_{i\in I}\plays{\G''_i}$ and  $\plays{\Nt}=\bigcup_{i\in I}\plays{\cml_i}\cup\bigcup_{i\in I}\plays{\Nt_i}$ by Lemma~\ref{pr}, so we get $\plays{\Sigma_{i\in I}\cml_i.\G''_i}=\plays{\Nt}$. Then we can use Rule \rn{TComm} to derive $\Sigma_{i\in I}\cml_i.\G''_i\vdash^{\mathcal P'}\Nt$.
\qed
\end{proof}

\section{Concluding Remarks, Related and Future Works}\label{rfw}

 In the setting of the {\em message-passing} communication model, it is possible to envisage 
mechanisms of interaction where a process can be, at the very same time, both a potential
sender and a potential receiver. In various frameworks  
for concurrent systems, like session types and 
communicating finite state machines, as well as the $\pi$-calculus, such mechanism is referred to as {\em mixed  choice}.
The flexibility and expressive power of mixed choice is understandably counterbalanced by a
difficult control of the behaviour of systems. 
That was arguably the motivation that mostly restrained the session type community from 
pursuing a thorough investigation of this sort of interactions. A stimulus in that direction has been instead recently given by some papers like~\cite{CMV22,PBK23,PBMK24,PY24b,PY24}. 
In particular,~\cite{CMV22,PY24b} investigate mixed choice 
for binary session types,  whereas~\cite{PY24} considers mixed choice in a MPST  setting~\cite{HYC08,Honda2016}  
following the approach of~\cite{ScalasY19}
(global types are in fact not taken into account in~\cite{PY24}). Even if the main concern  of~\cite{PY24}  is the expressivity of multiparty calculi 
(according to the full range of possible restrictions of mixed choice), 
type systems  assigning local types to processes are provided, 
where various predicates on contexts of local types are investigated.  
In~\cite{PBK23,PBMK24}  binary sessions with timeout and mixed choice are enriched with 
a semantics guaranteeing Progress and a type system enjoying Subject Reduction. 

Inspired by~\cite{PY24}, we carry on an investigation on the use of mixed choice
for synchronous communications in the  setting of SMPS~\cite{DGD22,BDL22}.
In SMPS, global types are inferred for sessions, i.e. parallel compositions of named processes, the latter being an abstraction for both processes and local types usually considered in MPST. 
Our processes can use now mixed choice. 
Subject Reduction, Session Fidelity,
as well as Lock Freedom are ensured for typable sessions.
The most relevant aspect of our type system is that we look at sessions as implicitly composed 
by modules whose participants freely interact via unrestricted mixed  choice, 
whereas the inter-module communications can be more easily controlled by 
 allowing communications with only one participant.  
Such an approach does not discard a priori any session. 
In fact all sessions
-- even those developed without any specific modularisation in mind -- can be modularised 
in a less or more refined way:
from a single  large module comprising the whole session, to a set of modules made by single
participants. In the former case the typing is less effective, since the global type would
result in a complete interaction tree, whereas in the latter case the typing coincides with the
standard  SMPS (with no mixed choice). 
In particular, our type system is conservative since, for processes without mixed choice,  
it coincides with the type system of~\cite{BDL24}, which is  at present 
one of the  most expressive. 
For mixed  choice 
there is only the type system of~\cite{PY24}, which is modularised by predicates on local types. An interesting property ensured by  that  
type system is {\em safety}. 
Safety in~\cite{PY24} entails that the protocol for process interactions
is such that, when  a participant 
intends to perform an interaction, it nondeterministically  chooses  
among all the participants that can interact with it. 
Then there is a nondeterministic decision concerning who has to play
the role of the sender and who  of  the receiver. 
At that point, the  sender performs an internal choice among the available outputs. 
Typing then guarantees that the possibility of interaction does
not depend on the chosen output.
We consider, instead, a simplified synchronisation protocol where 
there exists a nondeterministic choice among all the participants that can actually interact
 and all the possible communication interactions.
As mentioned in Section \ref{tc}, 
this approach makes ``!'' and ``?''
just two complementary synchronisation actions. It is however possible to modify our type system
in order to guarantee safety of sessions as defined in~\cite{PY24} by requiring that 
	
	\Cline{\pP\pp{\q!\msg{\lambda}.\PP\pplus\PP'}\in\Nt\text{ and }\pP\q{\pp?\msg{\lambda'}.\Q\pplus\Q'}\in\Nt\text{  imply }\pp\msg{\lambda}\q\in\LL\Nt}
	
\noindent	 The safety condition only ensures that an output finds the corresponding input if the receiver offers some inputs for the sender. An alternative condition is 

\Cline{\pP\pp{\q!\msg{\lambda}.\PP\pplus\PP'}\in\Nt\text{ implies }\Nt\SLTS{\sigma\cdot\pp\lambda\pq}  \text{ for some }\sigma}

\noindent With this last condition we would type less sessions, for example we would 
not type 

\Cline{\pP\pp{\q!\msg{\lambda} + \pr!\msg{\lambda'}}\parN\pP\pr{\pp?\msg{\lambda'}}}

\noindent
 and the election example. The feature of this alternative condition is that the choice between outputs is internal, in agreement with an asynchronous implementation. 
In such a case it is worth remarking that  a restriction of  the subtyping relation used in~\cite{PY24}
would be implicitly entailed by our  typing.

As first pointed out in~\cite{GHH21}, the  naive 
extension of the original type system~\cite{Honda2016} to sessions where input choices have different senders is unsound. In fact one can type sessions which reduce to untypable and stuck sessions.  Suitable conditions ensuring a sound extension have been proposed both for the synchronous~\cite{GHH21} and asynchronous~\cite{MMSZ21,CDG22} communications. 
Notably, the type system proposed here does not have this problem.

We notice that modular sessions can be obtained by connecting independent sessions via gateways,
according to the PaI approach to system composition.
When composing several typable SMPS systems (through {\em compatible} interfaces) one gets
a typable  system \cite{BDGY23}. 
Although the presence of mixed  choice 
does not seem to be a major obstacle, it is unlikely a result like the one of  \cite{BDGY23} could be easily translated in the present context.
In fact, as shown in Example \ref{ex:mconn} 
we allow the presence of multiple connectors per module.
This possibility is actually a severe impediment to the safeness of PaI composition. In fact,
let us take the following two typable sessions:

\Cline{
\pP\pp{\pq!\msg{\lambda}} \parN \pP\pq{\pp?\msg{\lambda}}\qquad\qquad \pP\pr{\ps?\msg{\lambda}} \parN \pP\ps{\pr!\msg{\lambda}}
}

\noindent
Taking all the participants as interfaces and considering that there exists 
a typable connection policy among them~\cite{BDGY23}, the PaI composition would result in the
following untypable $\Set{\Set{\pp,\pq}, \Set{\pr,\ps}}$-modular session

\Cline{
\pP\pp{\pr?\msg{\lambda}.\pq!\msg{\lambda}} \parN \pP\pq{\pp?\msg{\lambda}.\ps!\msg{\lambda}} \parN \pP\pr{\ps?\msg{\lambda}.\pp!\msg{\lambda}} \parN \pP\ps{\pq?\msg{\lambda}.\pr!\msg{\lambda}}
}

\noindent
where all the processes begin with an input, so forming a deadlock. This problem was overcome in \cite{GY23}, in a setting using projections and without mixed choice, by means of a suitable extension of the syntax of global types.

 It is worth noticing that mixed choice 
 make PaI composition problematic even by allowing one connector only. 
For example, by composing using $\pp$ and $\pt$ the following typable sessions

\Cline{\begin{array}{c}
\pP\pp{\pq!\msg{\lambda}_1+\pr?\msg{\lambda}_2} \parN \pP\pq{\pp?\msg{\lambda}_1+\ps?\msg{\lambda}_3}\parN \pP\pr{\pp!\msg{\lambda}_2+\ps!\msg{\lambda}_4}\parN \pP\ps{\pq!\msg{\lambda}_3+\pr?\msg{\lambda}_4}\\
\pP\pt{\pu?\msg{\lambda}_1+\pv!\msg{\lambda}_2} \parN \pP\pu{\pt!\msg{\lambda}_1+\pw!\msg{\lambda}_3}\parN \pP\pv{\pt?\msg{\lambda}_2+\pw?\msg{\lambda}_4}\parN \pP\pw{\pu?\msg{\lambda}_3+\pv!\msg{\lambda}_4}
\end{array}
}  

\noindent 
we get the session 

\Cline{\begin{array}{c}
\pP\pp{\pt?\msg{\lambda}_1.\pq!\msg{\lambda}_1+\pr?\msg{\lambda}_2.\pt!\msg{\lambda}_2} \parN \pP\pq{\pp?\msg{\lambda}_1+\ps?\msg{\lambda}_3}\parN \pP\pr{\pp!\msg{\lambda}_2+\ps!\msg{\lambda}_4}\parN \pP\ps{\pq!\msg{\lambda}_3+\pr?\msg{\lambda}_4}\parN\\
\pP\pt{\pu?\msg{\lambda}_1.\pp!\msg{\lambda}_1+\pp?\msg{\lambda}_2.\pv!\msg{\lambda}_2}\parN \pP\pu{\pt!\msg{\lambda}_1+\pw!\msg{\lambda}_3}\parN \pP\pv{\pt?\msg{\lambda}_2+\pw?\msg{\lambda}_4}\parN \pP\pw{\pu?\msg{\lambda}_3+\pv!\msg{\lambda}_4}
\end{array}
} 

\noindent
which reduces to the stuck session

\Cline{
\pP\pp{\pt!\msg{\lambda}_2} \parN \pP\pt{\pp!\msg{\lambda}_1}
} 

PaI composition  provides  
also an intuitive justification for the  shape  
of connectors 
in our modularisation. In fact, by composing systems via interfaces with unrestricted mixed choice, most of the
communication properties, if any, are not preserved,  as shown by  the  previous example. 

As future work we plan to investigate PaI 
composition for sessions with mixed  choice 
and asynchronous communication,  taking inspiration from~\cite{PBK23,PBMK24}, where asynchronous communication for sessions with mixed choice is first modelled.   
We deem worth investigating the modular approach to session types 
also for standard  MPST, 
as well as for the  recent approaches to global types and projections 
devised in~\cite{MMSZ21,LiSWZ23}.

{\bf Acknowledgments} We wish to gratefully thank the  anonymous reviewers for their thoughtful and helpful comments.

\bibliographystyle{eptcs}
\bibliography{session}

\end{document}